\documentclass[12pt,letterpaper]{article}
\usepackage[english]{babel}  
\usepackage{geometry}
\RequirePackage[OT1]{fontenc}
\RequirePackage{graphicx}
\usepackage[margin=1cm]{caption}
\RequirePackage{amsthm,amssymb}
\RequirePackage[cmex10]{amsmath}
\RequirePackage[colorlinks,citecolor=blue,urlcolor=blue]{hyperref}
\usepackage[round,authoryear]{natbib}
\usepackage{graphicx}
\usepackage{listing}
\usepackage{amsfonts}
\usepackage{amsmath}		
\usepackage{bbm} 	
\usepackage{graphicx,subfigure}
\usepackage{amsthm}
\usepackage{booktabs}
\usepackage{accents}
\usepackage{xcolor}
\usepackage{tabu}
\usepackage{caption} 
\usepackage{verbatim}
\usepackage{multirow}
\usepackage{authblk}
\usepackage[titletoc]{appendix}
\usepackage{array}
\usepackage{float}
\newtheorem{theorem}{Theorem}
\newtheorem{corollary}{Corollary} 

\newtheorem{proposition}{Proposition}
\newtheorem{assumption}{Assumption}
\newtheorem{remark}{Remark}[section]
\theoremstyle{definition}

\newtheorem*{genericthm*}{\thistheoremname}
\newenvironment{namedthm*}[1]
{\renewcommand{\thistheoremname}{#1}%
\begin{genericthm*}}
{\end{genericthm*}}
\newcommand{\thistheoremname}{}
  
\usepackage{color}
\usepackage{array}
\definecolor{dkgreen}{rgb}{0.0,0.53,0.74}
\definecolor{gray}{rgb}{0.5,0.5,0.5}
\definecolor{mauve}{rgb}{0.58,0,0.82}

\newlength{\defaultbaselineskip}
\usepackage{scalerel,stackengine}
\stackMath
\newcommand\widecheck[1]{%
\savestack{\tmpbox}{\stretchto{%
  \scaleto{%
    \scalerel*[\widthof{\ensuremath{#1}}]{\kern-.3pt\bigwedge\kern-.3pt}%
    {\rule[-\textheight/2]{1ex}{\textheight}}
  }{\textheight}%
}{0.4ex}}%
\stackon[1pt]{#1}{\scalebox{-1}{\tmpbox}}%
}
\newcommand{\dd}{d}
\newcommand{\SR}{\mathbb{R}}
\newcommand{\SN}{\mathbb{N}}
\newcommand{\E}{\mathbb{E}}

\newcommand{\filtr}{\mathcal{F}}
\newcommand{\prob}{\mathbb{P}}
\newcommand{\indicator}{\mathbbm{1}}
\newcommand{\measure}{\mathbb{P}}
\newcommand{\sigfield}{\mathcal{F}}
\newcommand{\stoptime}{\mathbf{N}}
\newcommand{\cv}{\xi}

\newcommand{\logs}{Y}
\newcommand{\elogs}{X}
\newcommand{\drift}{b}

\newcommand{\jump}{J}

\newcommand{\pn}{H}
\newcommand{\lrv}{\gamma} 
\newcommand{\cs}{\kappa} 
\newcommand{\vov}{\varsigma} 
\newcommand{\bm}{W}  
\newcommand{\Prm}{\mu}  
\newcommand{\cmps}{\nu}  
\newcommand{\cmpsmsr}{F}  

\newcommand{\ARLH}{D}
\newcommand{\ARLHc}{d(\bar{c})}

\makeatletter
\newcommand*\rel@kern[1]{\kern#1\dimexpr\macc@kerna}
\newcommand*\widebar[1]{%
  \begingroup
  \def\mathaccent##1##2{%
    \rel@kern{0.8}%
    \overline{\rel@kern{-0.8}\macc@nucleus\rel@kern{0.2}}%
    \rel@kern{-0.2}%
  }%
  \macc@depth\@ne
  \let\math@bgroup\@empty \let\math@egroup\macc@set@skewchar
  \mathsurround\z@ \frozen@everymath{\mathgroup\macc@group\relax}%
  \macc@set@skewchar\relax
  \let\mathaccentV\macc@nested@a
  \macc@nested@a\relax111{#1}%
  \endgroup
}
\makeatother

\graphicspath{{figures/}}

\title{Real-Time Detection of Local No-Arbitrage Violations\footnote{We are grateful to Chun He for excellent research assistance and to participants at various conferences and seminars for useful comments and suggestions.}}
\author[1]{Torben G. Andersen}
\author[1]{Viktor Todorov}
\author[2]{Bo Zhou}
\affil[1]{Kellogg School, Northwestern University}
\affil[2]{Department of Economics, Virginia Tech}
\date{}

\begin{document}

\maketitle
\thispagestyle{empty}
\abstract{\noindent This paper focuses on the task of detecting local episodes involving violation of the standard It\^o semimartingale assumption for financial asset prices in \textit{real time} that might induce arbitrage opportunities. Our proposed detectors, defined as stopping rules, are applied sequentially to continually incoming high-frequency data. We show that they are asymptotically exponentially distributed in the absence of It\^o semimartingale violations. On the other hand, when a violation occurs, we can achieve immediate detection under infill asymptotics. A Monte Carlo study demonstrates that the asymptotic results provide a good approximation to the finite-sample behavior of the sequential detectors. An empirical application to S\&P 500 index futures data corroborates the effectiveness of our detectors in swiftly identifying the emergence of an extreme return persistence episode in real time.\\}



\pagenumbering{arabic}
\renewcommand{\arraystretch}{1.5}\setlength{\baselineskip}{6.5mm}

\clearpage

\section{Introduction}

The \textit{no-arbitrage principle} is central to modern asset pricing theory (\citet{ross1976arbitrage}). \citet{delbaen1994general} demonstrate, within a frictionless setting, that arbitrage opportunities are precluded if and only if the price process constitutes a semimartingale. The standard class of no-arbitrage price processes in financial economics is the It\^o semimartingale, which is a semimartingale with characteristics that are absolutely continuous in time. However, recent work document episodic violations of the It\^o semimartingale assumption that, absent transaction costs, might induce arbitrage opportunities. A prominent example is the \textit{gradual jump} identified by \citet{barndorff2009realized} and further studied by \citet{christensen2014fact}. It occurs when an apparent return jump, identified from lower frequency data, instead reflects a strongly drifting, yet (near) continuous price path, when observed at higher frequencies. A related phenomenon is the so-called \textit{flash crash}, where a sudden collapse in price is reversed rapidly; see, e.g., the work on the May 2010 events in the S\&P 500 e-mini futures market by \citet{kirilenko2017flash} and \citet{menkveld2019flash}.\footnote{In the presence of trading costs and uncertainty surrounding the data generating process, such episodes are not necessarily true arbitrage opportunities, see, e.g., the discussion in \cite{andersen2021volatility}.}

The ``explosive'' price paths characterizing such events are unlikely to be generated by an It\^o semimartingale. To accommodate these occurrences, alternative models have been developed for violation episodes, including the \textit{drift burst} model proposed by \cite{christensen2020drift} and the \textit{persistent noise} by \cite{andersen2021volatility}, as a stochastic generalization of the former. These models contain a parameter $\tau$, indicating a random point in time located within a neighborhood in which an It\^o semimartingale violation occurs. Such episodes typically involve turbulent market conditions with extreme realized volatility, raising concerns of evaporating liquidity and general market malfunction. Moreover, our standard measures for monitoring return volatility are potentially subject to large biases, when the semimartingale assumption is violated. Consequently, identification of the onset, indicated by $\tau$, as well as duration of the extreme return drift episode is of great interest for regulators, industry practitioners and academics. This is the goal of this paper.

The detection problem can be addressed from two distinct perspectives. The more common is the \textit{offline} approach where the researcher observes the full dataset, and then conducts a ``one-shot'' procedure to identify whether and when violations have occurred. The vast literature on ex-post detection of structural breaks in macroeconomic and financial time series data, including \citet{andrews1993tests}, \citet{bai1996testing}, \citet{bai1998estimating}, and \citet{elliott2014pre}, falls within this category. More recently, \citet{bucher2017nonparametric} develop inference procedures for a change point in the jump intensity parameter for a L\'evy process with high-frequency data following this approach.

A more challenging and, for practitioners, investors and regulators, arguably more relevant perspective is the real-time setting, where data arrive continuously, and one wishes to detect the violation in a timely manner. This objective has some resemblance to that of the recent \textit{now-casting} literature in macroeconomics, where the aim is to update the assessment of the current and future state of the economy as new data are received in real time. For early initial work on a formal statistical framework in this setting, see, e.g.,  \citet{evans2005we} and \citet{giannone2008nowcasting}, while corresponding work utilizing financial data can be found in, e.g., \citet{andreou2013should} and \citet{banbura2013now}. Another related macro-econometric literature is initiated by \citet{diba1988explosive} regarding the detection of macroeconomic bubbles and crises. This methodology is extended by \citet{phillips2011explosive} using a recursive procedure based on right-tailed unit root tests to detect and locate the origin and terminal dates for bubbles. A series of subsequent studies provide further theoretical modifications, e.g., \citet{phillips2011dating}, \citet{phillips2014specification}, and \citet{phillips2015testingb}, while empirical applications are provided by \citet{phillips2018financial} and \citet{phillips2015testinga}.

Our objective is closely aligned with the latter \textit{real-time} or \textit{online} detection procedures. From a technical perspective, our approach is rooted in the statistic literature on the sequential detection problem for a change point. This literature typically deals with the mean of i.i.d samples or with the drift within a continuous-time model featuring both a drift and a scaled standard Brownian motion component. The literature goes back to, at least, Abraham Wald's sequential analysis. His work inspired the widely-known CUSUM rule by \citet{page1954continuous} and subsequently the Bayesian rule developed by \citet{shiryaev1963optimum} and \citet{roberts1966comparison}.\footnote{For more studies and optimality results, see, e.g., \citet{moustakides1986optimal}, \citet{ritov1990decision}, \citet{shiryaev1996minimax}, and \citet{moustakides2004optimality}. See also \citet{srivastava1993comparison} for a comparison and \citet{figueroa2019change} for the case of a L\'evy process. On the financial econometrics side, \citet{andreou2006monitoring} accommodates the CUSUM test for strongly dependent financial time series data, and \citet{andreou2008quality} applies real-time detection procedures to the structural parameters of credit risk models by monitoring the associated Radon-Nikodym derivative process.} We follow the former rule, but also deviate substantially, because that statistic requires knowledge of the alternative measure after the change, which is not a natural assumption in our setting. Specifically, our detector is based on the generalized likelihood ratio (GLR) statistic, which profiles out the unknown alternative parameter. For studies on this GLR-CUSUM procedure under an i.i.d setting, see \citet{siegmund1995using}, \citet{pollak1975approximations}, \citet{lai1979nonlinear}, and \cite{lai1995sequential}, among others.

Our paper differs from the aforementioned work in two key aspects. First, we rely on infill asymptotics and exploit the feature of accessing an asymptotically increasing number of observations of the process locally. This helps us formalize the notion of rapid detection of local It\^o semimartingale violations. These deviations from the semimartingale dynamics are only local in nature, unlike the earlier sequential detection literature which deals with detection of a permanent change. Second, we abandon the use of size versus power to characterize the properties of our detection procedure because we, by necessity, must apply our tests sequentially. If a test with fixed critical value is applied sequentially on a constant flow of newly arriving data, the null will inevitably be rejected repeatedly, and the traditional type I error literally explodes (with probability approaching one). To address this issue, we follow the sequential detection literature to design and evaluate the performance of our procedure using alternative metrics: the \textit{average run length (ARL)} and \textit{false detection rate (FDR)} or, more comprehensively, an asymptotic probability \textit{bound on the false detection (BFD)} rate versus a corresponding \textit{bound on the detection delay (BDD)} (see \cite{lorden1971procedures} and \citet{lai1995sequential}). Specifically, for the null probability measure, when there is no It\^o semimartingale violation, ARL measures the expected sample size until the first false detection, while FDR measures the probability of a false detection within a given period (and we refer to this period as BFD). For the alternative, BDD provides a probability bound on the number of observations before we achieve successful detection following a violation.\footnote{The sequential detection literature typically measures the behavior under the alternative via the {\it Expected Detection Delay}, or EDD. Characterizing the behavior of the latter is more complicated within our high-frequency setting. Thus, we focus on the probability estimate BDD instead.} Consequently, one strives to develop a detector that, conditional on a reasonably large ARL/small FDR, achieves the smallest possible BDD. An alternative would be to develop a time-varying boundary function -- in contrast to a constant threshold -- to control test size uniformly, see, e.g., \citet{chu1996monitoring}. Unfortunately, in a real-time sequential setting, such procedures struggle with the detection of structural changes that arrive late within the period. There is work seeking to alleviate this drawback, e.g., \citet{leisch2000monitoring}, \citet{horvath2004monitoring}, \citet{aue2004delay}, \cite{aue2006change}, and \cite{horvath2007sequential}. However, these procedures still tend to generate a significant detection delay, and we do not pursue this direction here.

As noted, our objective is to design statistical devices that detect local It\^o semimartingale violations swiftly and reliably after their occurrence, using continually incoming high-frequency data. Towards this end, we propose GLR-CUSUM type detectors as stopping times based on an estimated Brownian motion component of the latent asset price, which we recover from high-frequency return data along with short-dated options. We first establish the accuracy of this Brownian motion estimator, exploiting the option-based spot volatility estimate of \citet{todorov2019nonparametric} to standardize the high-frequency returns. This approach avoids complications stemming from the inconsistency of standard volatility measurements from returns, caused by local It\^o semimartingale violations, and it exploits the efficiency offered by option data for recovering spot volatility. Next, following the asymptotic setting in the sequential detection literature, we show that our detector is approximately exponentially distributed under the null (Theorem~\ref{thm:GLRw-CUSUM_null}) and develop a bound for the BDD under the alternative (Theorem~\ref{thm:GLRw-CUSUM_alter}). The former result implies that the ARL is the mean and FDR a percentile of the distribution. In turn, the latter result implies that our detectors achieve immediate detection of It\^o semimartingale violations under infill asymptotics. A Monte Carlo experiment finds that our theoretical results provide good guidance for the finite-sample properties of the detectors under a realistically calibrated simulation setting.

Finally, we apply our detection methods empirically to one-minute S\&P 500 equity (SPX) index futures data to investigate the pervasiveness of the It\^o semimartingale violation phenomenon and to assess whether our detector provides timely alerts regarding the failures. We further extend the procedure to obtain an identification rule for the duration of the violation. It is defined as the union of the time intervals on which our GLR-CUSUM statistic exceeds the chosen threshold. With a threshold specification for our detector leading to about 6\% daily FDR, we find these violations to be common in the SPX data --- with a little more than 1,000 such violations across 3,500 days. That said, more than half of the violations last for less than $10$ minutes, and only a small proportion exceeds $20$ minutes. Visual inspection suggests that, in most instances, our procedure detects such incidents within just a few minutes of their occurrence.

The remainder of the paper is organized as follows. In Section~\ref{sec:setup} we describe the setting and formulate the problem. In Section~\ref{sec:theory} we introduce our real-time detectors and study their asymptotic properties. Next, we carry out a Monte Carlo study in Section~\ref{sec:simulation} to illustrate the finite-sample performances of the detectors, and we then use them in an empirical application in Section~\ref{sec:empirical}. Section~\ref{sec:conclusion} concludes.

\section{Setup} \label{sec:setup}

The observed log-price process, $\logs_t$, is defined on a filtered probability space $\left(\Omega, \sigfield, (\sigfield_t)_{t \geq 0}, \measure\right)$. We assume that it may be decomposed as,
\begin{align} \label{eqn:observedprice_logs}
\logs_t ~=~ \elogs_t \,+\, \pn_t \, ,
\end{align}
where $\elogs$ represents an underlying efficient and arbitrage-free log-price. In addition, the observed price is contaminated by the noise component $\pn$, which is the source of brief episodes characterized by extreme return persistence, such as the so-called gradual jumps and flash crashes. Sections \ref{sec:setup_scheme} and \ref{sec:setup_effp} introduce our fairly standard specification for the observation scheme and $\elogs$, while Section \ref{sec:setup_noise} provides a detailed account of the dynamics for the noise component $\pn$.

\subsection{The Observation Scheme} \label{sec:setup_scheme}
We assume the price is observed over a given time interval $[0, T]$ at equidistant times, $t_i = i\Delta_n$, for all $i = 0, 1, \dots, nT$, where $\Delta_n = 1/n$ is the increment length. Let the $i$-th high-frequency log-return be denoted by,
\begin{align}
\Delta^n_i \logs \,=\, \logs_{i\Delta_n} \,-\, \logs_{(i-1)\Delta_n}.
\end{align}
Without loss of generality, we normalize the trading day to unity. Hence, $T$ refers to the number of trading days, and $n = 1/\Delta_n$ (assumed integer) is the number of observations per (trading) day. The time span of the data, $T$, is assumed fixed throughout.

Following the standard infill asymptotic framework, we let the sampling frequency, $\Delta_n$, shrink to $0$ (or, equivalently, $n \to \infty$). This equidistant sampling scheme can readily be relaxed to a non-equidistant one, requiring $\max_{\{i\in{1,\dots,n T\}}}(t_i - t_{i-1}) \to 0$.

\subsection{The efficient log-price $\elogs$} \label{sec:setup_effp}
The efficient log-price process $\elogs$ is an It\^o semimartingale,
\begin{align} \label{eqn:Itoprocess_elogs}
\elogs_t ~=~ \elogs_0 \,+\, \int_0^t \drift_s \, ds \,+\, \int_0^t \sigma_s \, d\bm_s \,+\, \int_0^t\int_{\SR} \delta(s,x) \, \Prm(ds,dx) \, ,
\end{align}  
where the initial value $\elogs_0$ is $\filtr_0$-measurable, the drift $\drift_t$ takes value in $\SR$, $\bm = (\bm_t)_{t \geq 0}$ is a one-dimensional standard Wiener process, $\delta: \, \SR_{+}\times\SR\mapsto\SR$ \, is a predictable mapping, and $\Prm$ is a Poisson random measure on $\SR_{+}\times\SR$ with predictable compensator (or intensity measure) $\cmps(dt,dx)=dt\otimes\cmpsmsr(dx)$. 

We impose the following assumption on the efficient price process.

\smallskip
\begin{assumption} \label{assumption:Itoprocess}
There exists arbitrary small $\mathbb{T}>0$ such that for the process $\elogs$ in equation (\ref{eqn:Itoprocess_elogs}), we have, for $t\in[0,\mathbb{T}]$ and $\forall p\geq 1$,
\begin{equation}
\mathbb{E}_0|b_t| \,+\, \mathbb{E}_0\left(\int_{\mathbb{R}}|\delta(t,x)|\cmpsmsr(dx)\right) \,+\, \mathbb{E}_0|\sigma_t|^p ~<~ C_0(p),
\end{equation}
where $C_0(p)$ is $\mathcal{F}_0$-adapted random variable that depends on $p$, and further
\begin{equation}
\mathbb{E}_0|\sigma_t-\sigma_s|^2 ~ \leq ~ C_0 \, |t-s|,~~~~\textrm{for $s,t\in[0,\mathbb{T}]$},
\end{equation}
where $C_0$ is $\mathcal{F}_0$-adapted random variable.
\end{assumption}
\smallskip

A few comments about this assumption are warranted. First, since our testing procedure concerns the asset price dynamics during an It\^o semimartingale violation initiated close to time zero, the assumption focuses on the behavior of $X$ in the vicinity of $0$. Second, the first part of Assumption~\ref{assumption:Itoprocess} involves the existence of conditional moments, but since $\mathbb{T}>0$ can be arbitrary small, these moment conditions are relatively weak. Third, the jump part of $X$ is modeled as an integral against a Poisson measure and, therefore, we restrict attention only to finite variation jumps. Finally, the ``smoothness in expectation'' condition for $\sigma$ will be satisfied whenever $\sigma$ is an It\^o semimartingale, which is the standard way of modeling stochastic volatility. We note, however, that this condition will not hold for the class of rough stochastic volatility models.\footnote{For this class of models, the right-hand side of the second equation in Assumption~\ref{assumption:Itoprocess} will be replaced by $C_0 \, |t-s|^{2H}$, for $0<H<1/2$ capturing the degree of roughness of the volatility path.} 
 
\begin{remark}
We refer to the model (\ref{eqn:Itoprocess_elogs}) for the efficient asset price $X$, along with Assumption~\ref{assumption:Itoprocess}, as the standard It\^o semimartingale model since it embeds most specifications used in current existing work. We note, however, that there is a gap between a semimartingale, required to preclude the existence of arbitrage opportunities, and the standard It\^o semimartingale considered above. The gap includes models in which the semimartingale characteristics are not absolutely continuous in time as well as models with infinite variation jumps. We do not consider those in our analysis.  
\end{remark}
\begin{remark}
We can readily accommodate an empirically relevant extension of the standard It\^o semimartingale model to the case, where the efficient price path features discontinuities triggered by economic announcements at pre-specified points in time. Since the announcement times are known a priori, one may remove the associated price increments from the analysis and proceed as if $X$ follows a standard It\^o semimartingale. To minimize notational complexity, we do not formally introduce this extension. 
\end{remark}

\subsection{The persistent noise $\pn$} \label{sec:setup_noise}

Our main objective is to develop inference tools to detect episodic violations of the It\^o semimartingale assumption, which is almost universally imposed in standard asset pricing theory. From a practical perspective, the occurrences of gradual jumps or, in more extreme manifestations, flash crashes, are of particular interest because the associated drift burst in the returns can induce severe biases in the high-frequency measurement of the (efficient) return variation. We include such short-lived episodes of extreme return persistence in our setup through the \textit{Persistent Noise (PN)} model of \citet{andersen2021volatility}, which is designed explicitly to accommodate these types of empirically observed phenomena. The PN model is a stochastically extended version of the \textit{Drift Burst (DB)} model by \citet{christensen2020drift}, and the real-time detection procedures developed below applies for either specification. We initially consider a simplified setting with only a single episode involving semimartingale violations. 

The PN model initiates a violation episode at a random point, $\tau_n$, i.e., $H_t = 0$ for $t < \tau_n$ and in our analysis $\tau_n\downarrow 0$. At $\tau_n$, the efficient price may display a discrete jump, $\Delta X_{\tau_n} \ne 0$. If an efficient price jump occurs, it likely reflects the arrival of new information or unusual ongoing trading activity. To the extent this information is not common knowledge or not readily interpretable, a gap will emerge between the efficient price, reflecting rational processing of all relevant information, and the market price which is determined by the interaction of risk averse and merely partially informed agents.\footnote{For an in-depth discussion of this feature across a large cross-section of stocks, see \citet{AACH}.} In other words, the efficient price may deviate from the market price, implying that the noise component is active, $H_{\tau_n} \neq 0$. In fact, if the information is not immediately observed or processed by market participants, we have $H_{\tau_n} = \Delta H_{\tau_n} = - \Delta X_{\tau_n}$, implying $\Delta Y_{\tau_n} = \Delta X_{\tau_n} - \Delta X_{\tau_n} = 0$. That is, the price reaction is delayed. An intermediate case is where the news generates a partial response. For example, if $H_{\tau_n} = - \Delta X_{\tau_n}/2$, the price jump underreacts to the new information. Such events may trigger a period of intense information acquisition and price discovery, inducing rapid price adjustment. After initiation at $\tau_n$, the future path of the noise component is modeled through a specific functional form, $g(t), t \in [\tau_n,\overline{\tau}]$. Finally, we ensure that the PN event terminates at some random point, $\overline{\tau} > \tau_{n\,}$, so that we have $H(t) = 0$ for $t > \overline{\tau}$.

We are now in position to formally introduce the PN model.

\medskip
\noindent {\bf The Persistent Noise (PN) model}: For some non-negative sequence $\tau_n\downarrow 0$, we let,
\begin{align} \label{eqn:modelH_persistencenoise_I}
\pn_{t} ~=~  f(\Delta X_{\tau_n}, \eta_{\tau_n}) \, g\left(t\right), ~~ t\geq \tau_n,
\end{align}
where $\eta_{\tau_n}$ is an $\mathcal{F}_{\tau_n}$-adapted random variable, $f$ is a continuous and bounded function, and $g$ is given as, 
\begin{align} \label{eqn:modelH_persistencenoise_II}
g(s) \,=\, \left\{ 1 - \left(\frac{s-\tau_n}{\overline{\tau}-\tau_n}\right)^{\vartheta_{pn}} \right\} \, \indicator_{\{s\in[\tau_n,\overline{\tau}]\}}, ~~~~ \vartheta_{pn}\in(0,0.5),
\end{align}
for some $\mathcal{F}_{\tau_n}$-adapted random $\overline{\tau}>\tau_n$.
\smallskip

The strength of the local return drift in the vicinity of the PN initiation point, $\tau_n$, in equation (\ref{eqn:modelH_persistencenoise_II}) increases as $\vartheta_{pn}$ takes on lower values, while the duration of the event is controlled by the realization of $\overline{\tau}$. We also note that $g(s)$ is zero, indicating no It\^o semimartingale violation, until $\tau_n$. Finally, $\eta_{\tau_n}$ allows for a random response to the events triggering the permanent noise component, which may or may not be associated with a jump in the efficient price.

To illustrate the dynamics of the PN model, Figure~\ref{fig:pricepath} plots a simulated sample path for the efficient price $\elogs$ (blue) across the 6.5 hour trading day along with the associated path for the observed price $\logs$ (black). The efficient price path features a positive jump at $\tau_n = 0.25$ (at $97$ minute), but (informational) frictions prevent this from being recognized by market participants, so there is no instantaneous jump in the observed price. In terms of the model, we register a corresponding negative latent noise jump at $\tau_n$, namely $f(\Delta X_{\tau_n}, \eta_{\tau_n}) = - \Delta X_{\tau_n}$. This is a purely mechanical effect. If there is no observed price change, but the efficient price jumps, then the deviation between the two -- the noise term $\pn$ -- must offset the efficient price jump.

As noted, Figure~\ref{fig:pricepath} is designed to replicate the gradual jump phenomenon. It is evident that this PN episode also may reflect the recovery from a so-called flash crash that hits its nadir at $\tau_n$. Thus, one may replicate the typical flash crash pattern by combining this ``recovery phase'' with an inverted version of the PN path prior to $\tau_n$, see \citet{andersen2021volatility} for a more detailed discussion.

\begin{figure}[H] 
\hspace*{-0mm}
\centering
\includegraphics[width=6.5in]
{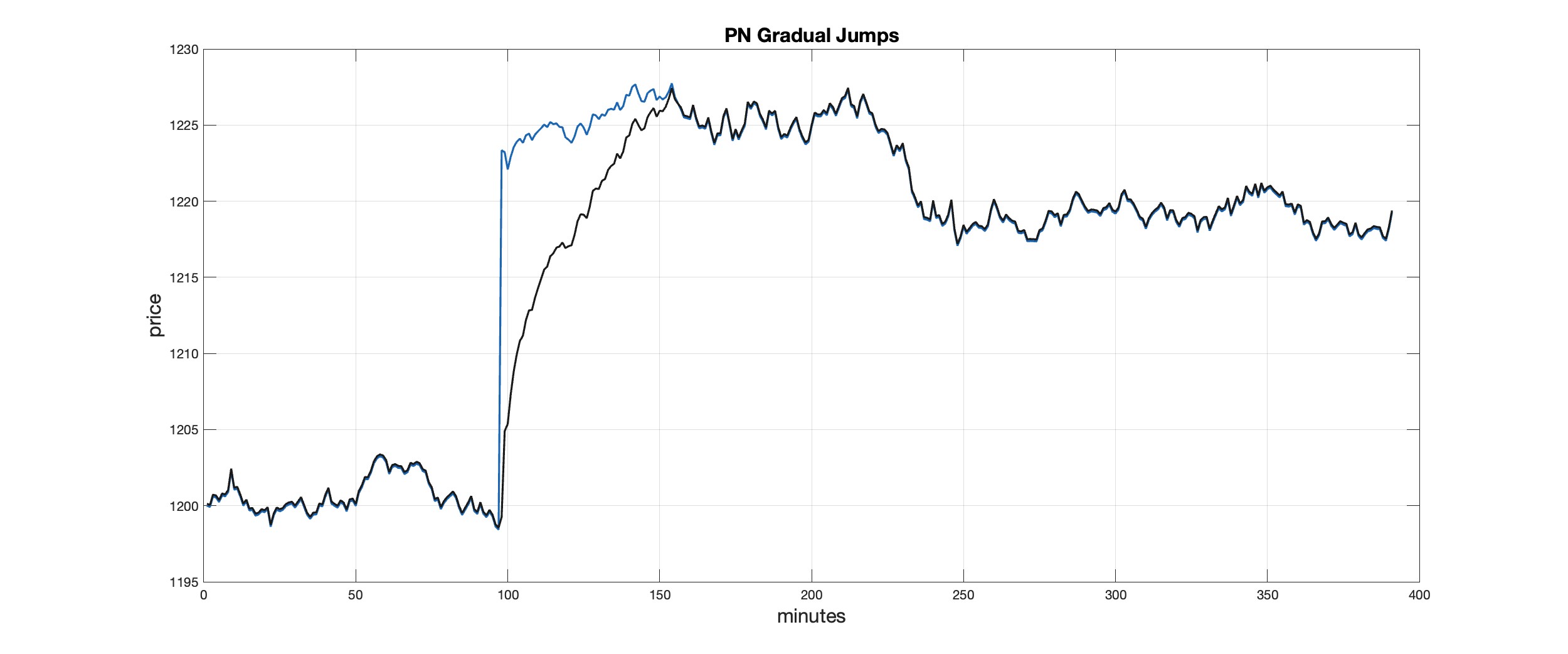}
\caption{Price path samples of the efficient price $\elogs$ (in blue) and observed price $\logs$ (in black) simulated by the Persistent Noise (PN) model.}
\label{fig:pricepath}
\end{figure}
\vspace{-0.1in}

From a technical perspective, it is convenient to introduce a version of the DB model by \citet{christensen2020drift}.\footnote{\citet{christensen2020drift} allow for simultaneous volatility and drift bursts. We do not include this feature, as we find no systematic evidence for significant elevation in option-based volatility measures at such times.}
\smallskip

\medskip
\noindent {\bf The Drift Burst (DB) model}: For some non-negative sequence $\tau_n\downarrow 0$, we have,
\begin{align} \label{eqn:modelH_driftburst}
\pn_t = \int_0^t c(s-\tau_n)^{-\vartheta}\indicator_{\{s\in[\tau_n,\overline{\tau}]\}} ds, ~~ \vartheta \in (0.5, 1),
\end{align}
for a constant $c$ and $\tau_n \leq \overline{\tau}$. $\tau_n$ and $\overline{\tau}$ have the same interpretation as in the PN model.
\smallskip

The DB model may be viewed as a differential version of the PN model by observing that $\int_{t_1}^{t_2}(s-\tau)^{-\vartheta}ds = (1-\vartheta)^{-1}(t_2-\tau)^{1-\vartheta} - (1-\vartheta)^{-1}(t_1-\tau)^{1-\vartheta}$ for $0\leq \tau \leq t_1<t_2$. Therefore, the asymptotic analysis will be identical for the PN model and the DB model with $\vartheta = 1-\vartheta_{pn}$. We will exploit this equivalence in our theoretical arguments.

Henceforth, we denote the probability and expectation by $\prob_{\tau_n}$ and $\E_{\tau_n}$ for the case where the observed log-price is generated as above, and by $\prob_{\infty}$ and $\E_{\infty}$ if there are no It\^o semimartingale violations (that is, when $\tau_n = \infty$).

\section{Real-Time Detection}
\label{sec:theory}

In this section, we propose GLR-CUSUM type detectors as stopping rules for the local It\^o semimartingale violations illustrated above. Next, we develop their theoretical properties. This involves characterizing the accuracy for our estimate of the Brownian motion driving the price, followed by approximation results for the distribution of the stopping rules, leading to analytic formulas for ARL and FDR, and finally a probability result for BDD which, in turn, implies a probability bound on the speed of detection.

\subsection{The GLR-CUSUM stopping rule}
Our local It\^o semimartingale violation detector is built upon a local estimator of the Brownian motion driving the price, defined as,
\begin{align} \label{eqn:bm_estimate_local}
\widehat{\bm}_{l_1,\,l_2}
~=~
\sum_{i = l_1+1}^{l_2} \,\frac{\Delta_i^n\logs}{~\hat{\sigma}_{(i-1) \, \Delta_n}~} \, \indicator_{\{|\Delta_i^n\logs|<\zeta\Delta_n^\varpi\}} \,,
\end{align}
for some $l_1,l_2\in\SN$ such that $0 \leq l_1 < l_2 \leq nT$, and $\hat{\sigma}_{i\Delta_n}$ is an option-based nonparametric spot volatility estimator defined later. As detection of the semimartingale violation inevitably is subject to some delay, the option-based $\hat{\sigma}$ enables us to avoid the bias in volatility estimation induced by the bursting drift following $\tau_n$ (see \cite{andersen2021volatility}).

Our CUSUM-type detector, employing the \textit{generalized likelihood ratio} (GLR) statistic, is now defined by the stopping rule,
\begin{align} \label{eqn:stoptime_GLRw}
\stoptime^{\rm w}(\cv) ~=~ \inf\left\{\, l > 0: \max_{l - w_n \leq k < l} \, \frac{|\widehat{\bm}_{k,l}|}{~\sqrt{(l-k)\, \Delta_n}~} ~>~ \cv \, \right\} \, ,
\end{align}
where $w_n$ controls the maximum window length for $\widehat{\bm}_{k,l}$, and $\cv$ is a chosen threshold.\footnote{The terminology of generalized likelihood ratio statistic here refers to (the square root of $2$ times) the likelihood of a Brownian motion with drift $\mu$ versus drift of $0$, $\mu \, \bm_{k,l} - \mu^2(l-k)\Delta_n/2$, with the unknown alternative $\mu$ being replaced by its maximum likelihood estimate $\widehat\mu = \bm_{k,l}/\sqrt{(l-k)\Delta_n}$.} In words, at time point $l$, we calculate the absolute value of the Brownian increment estimate $\widehat{\bm}_{k,l}$ over the interval $[k,l]$ (normalized by $1/\sqrt{(l-k)\Delta_n\,}$) for each point $k$ before $l$, and then take the maximum. If the maximum exceeds the threshold $\cv$, the detector reports an alarm, signaling detection of an It\^o semimartingale violation.

\subsection{Recovering the Driving Brownian Motion for the Price}
Our theoretical results concerning detection of local semimartingale violations hinge on the accuracy of the estimator $\widehat{\sigma}_{t}$. The following assumption is sufficient.

\smallskip
\begin{assumption}\label{assumption:sigma_est}
For some arbitrary small $\mathbb{T}>0$, we have $\inf_{t\in[0,\mathbb{T}]}\sigma_t>0$ and 
\begin{equation}
\mathbb{E}_{0\,} |\widehat{\sigma}_t - \sigma_t|^2 ~ \leq ~ C_0 \, \delta_n^2,~~~~t\in[0,\mathbb{T}],
\end{equation}
for some $\mathcal{F}_0$-adapted positive random variable $C_0$, and a deterministic sequence $\delta_n\rightarrow 0$.

In addition, we have $\widehat{\sigma}_t>C_0/l_n$, for $t\in[0,\mathbb{T}]$, some $\mathcal{F}_0$-adapted positive random variable $C_0$, and $l_n = \log(1/\delta_n)$.
\end{assumption} 
\smallskip

Several features of the above assumption are noteworthy. First, we impose conditions on the volatility estimator only in the vicinity of zero. Second, we require that $\sigma_t$ is non-vanishing in a neighborhood of zero. This is important for the behavior of the volatility estimator, but also for the overall validity of our detection procedure, as the diffusive component of $\elogs$ drives our asymptotic results. Third, the deterministic sequence $\delta_n$ captures the rate of convergence of the option-based volatility estimator. This rate is determined by the mesh of the strike grid of the available options as well as their tenor, see \cite{todorov2019nonparametric}. 

The above assumption is natural in the absence of a local semimartingale violation. Its validity is less clear if there is a violation of the type discussed in the previous section in the vicinity of $t=0$. However, note that a drift burst or persistent noise incident will not affect the norm of the conditional characteristic function of the price increment which forms the basis for the estimator of \cite{todorov2019nonparametric}. Hence, the consistency of our spot volatility estimator should not be affected by the local semimartingale violation. That said, the quality of the option-based estimator might worsen during such periods due to poor option data quality. To guard against this possibility, we impose some filters on the option data in our application.

To facilitate the statement of our next result, we introduce the following notation,
\begin{align*}
\widecheck{\bm}_{l_1,\,l_2} ~=~ \sum_{i = l_1+1}^{l_2} \Delta_i^n\bm,\end{align*}
for some $l_1,l_2\in\SN$ such that $0 \leq l_1 < l_2 \leq nT$. The issue in this section is how closely $\widehat{\bm}_{l_1,\,l_2}$ approximates $\widecheck{\bm}_{l_1,\,l_2}$. The key result is provided by the following proposition.

\smallskip
\begin{proposition} \label{prop:consistency_BM}
Suppose Assumptions~\ref{assumption:Itoprocess} and \ref{assumption:sigma_est} hold and $H\equiv 0$. Let $k_n\to\infty$ and $k_n\Delta_n\to 0$ as $\Delta_n \to 0$. Let $0\leq l_1^n<l_2^n$ be such that $l_1^n\Delta_n\rightarrow 0$ as $\Delta_n\rightarrow 0$ and, as before, $l_n = \log(1/\delta_n)$. Then, for a positive $\cv>0$, we have,
\begin{equation}
\begin{split}
&\mathbb{P}_{0}\left(\max_{|l_2^n-l_1^n|\leq k_n} \frac{~|\widehat{\bm}_{l_1^n,l_2^n} - \widecheck{\bm}_{l_1^n,l_2^n}|~}{\sqrt{(l_2^n-l_1^n) \, \Delta_n}} \, > \, \xi\right) \\ &~~~~~~~~~ \leq ~ C_0 \, k_n \, \left(\frac{~\sqrt{k_n} \, l_n \, \delta_n + \sqrt{k_n \, \Delta_n}~}{\xi} \,+\, \frac{k_n \, \Delta_n}{~1\wedge \xi^2~} \,+\, \Delta_n^{1-\varpi}\right) \, ,
\end{split}
\end{equation}
for some $\mathcal{F}_0$-adapted positive random variable $C_0$ that does not depend on $k_n$, $\Delta_n$ or $\xi$. 
\end{proposition}
\smallskip

\subsection{Average Run Length (ARL)}
We first analyze the behavior of the stopping time $\stoptime^{\rm w}(\cv)$ under the null measure, $\prob_{\infty}$, where there is no It\^o semimartingale violation, i.e., $\tau_n = \infty$. More specifically, our interest is in evaluating the ARL, formally defined as $\E_{\infty}[\, \stoptime^{\rm w}(\cv) \, ]$. We show in Theorem~\ref{thm:GLRw-CUSUM_null} below that $\stoptime^{\rm w}(\cv)$, under $\prob_{\infty}$, is approximately exponentially distributed, leading to an approximation result for ARL. Theorem~\ref{thm:GLRw-CUSUM_null} follows by adapting the theoretical framework in \citet{siegmund1995using}, based on the random fields analysis in \cite{siegmund1988approximate} to our setting (see also \citet{yao1993boundary}). We use the symbol ``$x_n \sim y_n$'' to denote ``$x_n$ is asymptotic equivalent to $y_n$'', that is, $\lim_{n\to\infty} x_n/y_n = 1$.

\smallskip
\begin{theorem} \label{thm:GLRw-CUSUM_null}
Denote the standard normal CDF and PDF by $\Phi$ and $\phi$. Let $w_n$ be such that $w_n \sim a_n \, \cv^2$, with $a_n$ being a deterministic sequence converging to $a\in (0,\infty]$. Suppose the conditions in Proposition~\ref{prop:consistency_BM} hold. Then, if $(l_n \, \delta_n + \sqrt{\Delta_{n\,}})/\, [\cv \, \phi(\cv)] \rightarrow 0$, as $\cv\to\infty$, $\stoptime^{\rm w}(\cv)$ will be asymptotically exponentially distributed with expectation,
\begin{align} \label{eqn:ARL_thy_GLRw}
\E_{\infty}[\, \stoptime^{\rm w}(\cv) \, ] ~ \sim \, \frac{1}{~ \ARLH_a \, \cv \, \phi(\cv)~ } \,\, .
\end{align}
Here $\ARLH_a$ is a positive constant depending on $a$ such that $\ARLH_a\to\ARLH$, as $a\to\infty$, where $\ARLH = \int_0^\infty x \, \nu^2(x) \, dx$ \, with \, $\nu(x) = 2x^{-2} \exp\left[-2\sum_{n=1}^{\infty}n^{-1} \, \Phi(-x\sqrt{n}/2)\right]$ \, for $x>0$.
\end{theorem}
\smallskip

The proof of Theorem~\ref{thm:GLRw-CUSUM_null} is provided in Appendix~\ref{appsec:proofs}. Following \citet{siegmund1995using}, it is useful for numerical evaluation to apply the approximation $\nu(x) = \exp(-\rho x) + o(x^2)$ for $x \to 0$, where the value of the constant $\rho$ is about $0.583$. Based on this, we can numerically determine $\ARLH$ to take the value $0.735$.

We desire for the ARL, under the null, to be as large as possible or, equivalently, $\cv$ to be as large as possible. The optimal such $\cv$, up to a log term and for a given $w_n$, will be $\cv$ such that $\phi(\cv) \sim (\delta_n \, l_n \vee \sqrt{\Delta_n})$. It is reasonable to assume that the size of the error in estimating spot volatility from options is much smaller than the high-frequency approximation error, that is, $\delta_n \, l_n/\sqrt{\Delta_n} \rightarrow 0$. Under this condition and the indicated choice of $\cv$, we obtain from Theorem~\ref{thm:GLRw-CUSUM_null} that,
\begin{equation}
\E_{\infty}[\, \stoptime^{\rm w}(\cv) \, ] ~ \sim \, \frac{1}{~\sqrt{\Delta_n}~} ~ \sqrt{\, \log\left(\frac{1}{\Delta_n}\right) \, } \, .
\end{equation}

We note that this ARL value can be obtained with $w_n$ (the window size) taking values within a wide range --- from an order of $\cv^2$ through using (almost) all available observations on a given day (retaining $w_n \, \Delta_n \rightarrow 0$). For the case of $w_n/ \, \cv^2 \rightarrow \infty$, we may denote our detector $\stoptime(\cv)$, and replace $\ARLH_a$ by $\ARLH$, as $\ARLH_a \to \ARLH$ with $a \to \infty$.\footnote{For the standard i.i.d.\ case, more specific guidelines in selecting $w_n$ can be found in \citet{lai1995sequential}.}

\subsection{False Detection Rate (FDR)}

Given the nature of our sequential detection procedure, we cannot rely on the regular notion of test size to control the rejection rate under the null hypothesis. The preceding subsection instead focuses on the concept of ARL, establishing this quantity (the ARL) as the mean value of the asymptotic distribution for stopping time $\stoptime^{\rm w}(\cv)$ under $\prob_{\infty}$. However, this distributional result has wider implications. For example, it allows us to control the false detection rate (FDR) as,
\begin{align} \label{eqn:FDR}
\prob_{\infty \,} [\,\stoptime^{\rm w}(\cv) \, < \, \ell_{n}\,] ~ \leq ~ \alpha \, ,
\end{align}
where $\ell_n$ indicates an asymptotically increasing sequence of observations, while $\alpha$ is the bound for the FDR -- the probability of triggering (at least) one false alarm after observing no more than $\ell_n$ price increments. If $\alpha$ is asymptotically shrinking, then we refer to $\ell_n$ as the \textit{bound (in probability) on false detection} or BFD. 

The FDR is often viewed as a more robust measure than ARL. Typically, the objective is to ensure an expected long duration (ARL) before a false alarm is triggered. The problem is that a stopping rule may feature a large ARL and still retain a high probability of false alarms within short periods. Although this does not apply in our case, we still proceed with a result concerning the limiting behavior of FDR under the null, since it provides a guideline for selecting a threshold in empirical applications. Moreover, it renders the Monte Carlo studies more computationally tractable, as we avoid simulating an excessively long sample to assess the ARL, especially when we experiment with large thresholds. Below, we use the symbol ``$x_n \lesssim y_n$'' to denote ``$x_n$ is asymptotic equivalent to or less than $y_n$'', that is, $\lim_{n\to\infty} x_n/y_n \leq 1$. 

\smallskip
\begin{corollary} \label{cor:GLR_FDR}
Suppose the conditions in Theorem~\ref{thm:GLRw-CUSUM_null} hold and let $\ell_n$ be such that $\ell_n \, \cv \, \phi(\cv) \to 0$, as $\cv\to\infty$. Then, we have,
\begin{align} \label{eqn:GLRw_FDR}
\prob_{\infty}[\, \stoptime^{\rm w}(\cv) < \ell_{n\,}] ~  \lesssim ~ \ell_n \, \ARLH_a \, \cv \, \phi(\cv) \, .
\end{align}
\end{corollary}
\smallskip

This corollary follows directly from Theorem~\ref{thm:GLRw-CUSUM_null}. Specifically, the approximate exponential distribution of $\stoptime^{\rm w}(\cv)$ indicates $\prob_{\infty}[\stoptime^{\rm w}(\cv) < \ell_n] \sim 1 - \exp(-\ell_n \, \ARLH_a \, \cv \, \phi(\cv))$, and the result then follows from the inequality $1 - \exp(-x) \leq x$ for $x$ close to $0$. A sequence $\ell_n$ satisfying the conditions of Corollary~\ref{cor:GLR_FDR} constitutes a BFD.

\subsection{Bound on Detection Delay (BDD)}

When an It\^o semimartingale violation occurs at some time $\tau_n\downarrow 0$, we wish to detect it as quickly as possible, using discrete high-frequency observations of $Y$ starting at time 0. We refer to a deterministic sequence $T_n\rightarrow\infty$ such that $\mathbb{P}\left( \, \stoptime^{\rm w}(\cv) \, > \, \tau_n/\Delta_n+T_n \, \right) ~\rightarrow~ 0$, when $H$ is nontrivial, as a \textit{bound (in probability) on detection delay} or BDD. We want $T_n$ as small as possible. Theorem~\ref{thm:GLRw-CUSUM_alter} below provides  a lower 
bound on BDD. 

\smallskip
\begin{theorem} \label{thm:GLRw-CUSUM_alter}
Suppose $b_s = b_0$, $\sigma_s = \sigma_0$ and $\Delta X_s = 0$, almost surely, for $s$ in a neighborhood of zero and that Assumption~\ref{assumption:sigma_est} holds. Let $H$ be given by the drift burst model in equation (\ref{eqn:modelH_driftburst}) for some $c\neq 0$. Let $w_n$ be such that $w_n \sim a_n \, \cv^2$, with $a_n$ being a deterministic sequence converging to $a\in (2/c^2,\infty]$. Finally, let $\varpi\in (0,1/2)$ and $\vartheta\in(1/2,1)$ as well as $\cv\Delta_n^{\iota}\rightarrow 0$ for any $\iota>0$. We then have,
\begin{equation}\label{eq:BDD}
\mathbb{P}\left( \, \stoptime^{\rm w}(\cv) \, > \, \tau_n/\Delta_n+T_n \, \right) ~\rightarrow~ 0 \, ,
\end{equation}
for any sequence $T_n\rightarrow\infty$ such that $T_n \, \Delta_n\rightarrow 0$,\, $w_n/T_n\rightarrow 0$, \, $T_n/\Delta_n^{(1/2-\vartheta)/\vartheta}\rightarrow\infty$ \, and $\sqrt{w_n} ~ l_n \, \delta_n \, \Delta_n^{\varpi-1/2} / \, (\cv \, T_n) \rightarrow 0$.
\end{theorem}
\smallskip

We invoke a number of simplifying assumptions in the derivation of Theorem~\ref{thm:GLRw-CUSUM_alter}. Primarily, we assume that $X$ is continuous and features constant drift and constant diffusive volatility in a (shrinking) neighborhood of zero. This style of assumption is common in the high-frequency financial econometrics literature and --- as is typically the case --- it should be possible to relax at the cost of more complicated derivations.

Theorem~\ref{thm:GLRw-CUSUM_alter} is proved by constructing an alternative measure, say $\widetilde\prob$, under which the semimartingale violations are milder. Specifically, the expected value of the increments under $\widetilde\prob$ are smaller in asymptotic order of magnitude than under $\mathbb{P}$. Consequently, it takes longer to detect the change, and thus the BDD under $\widetilde\prob$ is an upper bound for that under $\mathbb{P}$.

The object characterized by the BDD condition (\ref{eq:BDD}) is a diverging sequence of observations, $(T_n)$, sufficiently large to ensure almost sure detection, in our infill asymptotic setting, of an It\^o semimartingale violation. It is natural to require that our detector identifies such violations before triggering a false alarm with probability approaching one or, equivalently, BDD $<$ BFD. From Theorem~\ref{thm:GLRw-CUSUM_alter}, we need $T_n/\Delta_n^{(1/2-\vartheta)/\vartheta}\rightarrow\infty$, while BFD, from Corollary~\ref{cor:GLR_FDR} and the discussion thereafter, is of asymptotic order $1/\, (\cv \, \phi(\cv))$. Thus, we must have $\cv \, \phi(\cv) \, \Delta_n^{(1/2-\vartheta)/\vartheta} \to \, 0$ to ensure BDD $<$ BFD. This condition may be satisfied for any $\vartheta<1$, provided $\cv$ is chosen as large as possible, while still guaranteeing that Theorem~\ref{thm:GLRw-CUSUM_null} applies (recall the discussion after Theorem~\ref{thm:GLRw-CUSUM_null}). 

The BDD embeds two sources of delay (as shown in our proof): (i) the truncation of price increments to guard against jumps; (ii) the requisite cumulation of non-truncated returns in our detection statistic.
Asymptotically, the first delay term dominates the second. Nonetheless, in finite samples, the BDD will be determined by both effects. 

\smallskip
\begin{remark}
We can also introduce a double-window-limited GLR rule, defined as,
\begin{align} \label{eqn:stoptime_GLRww}
\stoptime^{\rm ww}(\cv) := \inf\left\{l > 0: \max_{l - w_n - r_n \leq k < l - r_n}\frac{|\widehat{\bm}_{k,l}|}{\sqrt{(l-k)\Delta_n}} > \cv \right\},
\end{align}
where $r_n$ imposes a restriction on the minimum span for the statistic $\widehat{\bm}_{k,l}$. This helps reduce the false detection error under the null induced by extreme values for the statistic $\widehat{\bm}_{k,l}$ due to only a few ($\, < r_n$) increments. Consequently, this will generate a longer ARL/smaller FDR under the null, and a slightly larger BDD under the alternative. We do not analyze $\stoptime^{\rm ww}(\cv)$ theoretically, but evaluate its performance numerically via simulation. The Monte Carlo study in Section~\ref{sec:simulation} shows that setting $r_n$ to, say, $5$, we obtain effective protection against outliers without affecting the detection delay by much.

We note that the literature on testing for bubbles in macroeconomic settings often imposes a similar  minimum duration restriction on the bubble period (see, e.g., \citet{phillips2011explosive}, \citet{phillips2011dating}, and \citet{phillips2015testingb}).
\end{remark}

\section{Simulation Study} \label{sec:simulation}

In this section we carry out a simulation study to assess the finite-sample performance of our procedure for swift detection of local It\^o semimartingale violations.

\subsection{Simulation setting} \label{subsec:simu_setting}

To generate sample paths for the efficient log-price $\elogs$, we simulate the following Heston type stochastic volatility model with jumps, 
\begin{align}
\dd\elogs_t    &~=~ \drift_t \, \dd t \,+\, \sigma_t \, \dd\bm_{1,t} \,+\, \dd\jump_t \, , \\
\dd\sigma_t^2 &~=~ \cs \, (\lrv - \sigma_t^{2\,}) \, \dd t \,+\, \vov \, \sigma_t \, \bm_{2,t}
\end{align}
where $\bm_{1}$ and $\bm_{2}$ are standard Brownian motions with correlation $\E(\dd\bm_{1,t}, \dd\bm_{2,t}) = \rho \, \dd t$, and $\jump_t$ denotes the jump term for which we employ a Poisson process with intensity $p_\elogs$ and jump size distribution $\mathcal{N}(0, \lambda_\elogs^2)$\,.

In terms of parameter specification, we set the initial efficient log-price $\elogs_0$ to $\log(1200)$, the drift term $\drift_t$ to zero, and the unit of time to one trading day. The volatility process $\sigma_t^2$ is initiated at its unconditional mean of $\lrv$ on day one while, for other days, it is initiated at the ending value of the previous day. The annualized parameter vector for the Heston model is set to $(\cs, \lrv, \vov, \rho) = (5, 0.0225, 0.4, -\sqrt{0.5})$. For the jump components, we let $p_\elogs = 3/5$, corresponding to 3 jump per week on average, and $\lambda_\elogs = 0.5\%$. 

For the It\^o semimartingale violation term $\pn$, we employ the DB and PN models (\ref{eqn:modelH_driftburst}) and (\ref{eqn:modelH_persistencenoise_I})--(\ref{eqn:modelH_persistencenoise_II}), respectively, to generate gradual-jump type patterns. Specifically, for the DB model, we set $c = 3$, $\tau_n = 0.25$ and $\overline\tau = 0.4$, and explore $\vartheta \in \{0.55, 0.65, 0.75\}$. For the PN model, we add a jump in $\elogs$ at $\tau_n = 0.25$ each day. We let $f(\Delta\elogs_\tau, \eta_\tau) = \eta_\tau \, \Delta\elogs_\tau$ with $\eta_\tau = -1$, $\overline{\tau} = 0.4$, and explore $\vartheta_{pn} \in \{0.45, 0.35, 0.25\}$, corresponding to $\Delta\elogs_\tau = \{1.4\%, 2.0\%, 3.0\%\}$. Finally, $\widehat\sigma_t = \sigma_t\times(1 + 0.02 \times Z)$, where $Z$ is a standard normal random variable, serves as proxy for the (noisy) option-based volatility estimate.

\subsection{Simulation results} \label{subsec:simu_results}

We first present results for the null measure $\prob_\infty$, i.e., no It\^o semimartingale violations.

\begin{table}[!htb] 
	\centering
	\small
	\caption{Average Run Length (ARL)}
\begin{tabu}{c | c c | c c | c}
\hline\hline
& \multicolumn{2}{c}{$\stoptime(\cv)$} & \multicolumn{2}{c}{$\stoptime^{\rm w}(\cv)$} & $\stoptime^{\rm ww}(\cv)$ \\
\hline
Threshold $\cv$ & Theory & Simulation & Theory & Simulation & Simulation  \\ 
\hline 
3.4 & ~358 & ~343 & ~393 & ~367 & ~684 \\ 
3.5 & ~487 & ~475 & ~536 & ~509 & ~939 \\  
3.6 & ~669 & ~673 & ~740 & ~718 & 1298 \\
3.7 & ~931 & ~956 & 1033 & 1014 & 1792 \\    
3.8 & 1310 & 1312 & 1457 & 1432 & 2456 \\   
3.9 & 1864 & 1881 & 2082 & 2058 & 3356 \\  
4.0 & 2682 & 2609 & 3007 & 2897 & 4406 \\      
\bottomrule
\end{tabu}
\caption*{The ARL is based on simulated $1$-minute data for 1,000 replications. We choose the window size $w_n = 30$ minutes for $\stoptime^{\rm w}(\cv)$ and $(w_n, r_n) = (30, 5)$ minutes for $\stoptime^{\rm ww}(\cv)$. We set the jump truncation threshold $\zeta = 4\,\widehat{\sigma}_{t-1}^{\rm \, med}$ for all cases.}
\label{table:ARL}
\end{table}

In Table~\ref{table:ARL}, we report the ARL values for our proposed detectors $\stoptime(\cv)$ (with $w_n = \infty$), $\stoptime^{\rm w}(\cv)$ ($w_n = 30$ minutes), and $\stoptime^{\rm ww}(\cv)$ ($(w_n, r_n) = (30, 5)$ minutes) based on our simulation setting. We also provide the theoretical ARL values for the former two detectors based on Theorem~\ref{thm:GLRw-CUSUM_null}. Comparing the first two and the next two columns we find that, for a wide range of threshold values ranging from $3.4$ to $4.0$, the Monte Carlo values for $\stoptime(\cv)$ and $\stoptime^{\rm w}(\cv)$ are close to their corresponding theoretical values, indicating that the asymptotic approximation (\ref{eqn:ARL_thy_GLRw}) captures the finite-sample performance well. When comparing the ARLs of $\stoptime(\cv)$ to those of $\stoptime^{\rm w}(\cv)$, we find that the latter are slightly larger than the former. This indicates that the window limit $w_n$ does not induce any major deterioration under the null, even with the relatively small value (here, $30$ minutes). On the contrary, in the last column, we find the ARL values for $\stoptime^{\rm ww}(\cv)$ to be significantly larger than the previous ones within the same settings. Evidently, the window limit $r_n$ has a significant impact under the null measure by ignoring false detections induced by just a couple of increments. Thus, there is room to experiment along this dimension in order to improve the FDR, although it is critical to also monitor the associated increase in the detection delay under the alternative.

\begin{table}[!htb] 
	\centering
	\small
	\caption{False Detection Rate (FDR)}
\begin{tabu}{c | c c | c c | c}
\hline\hline
& \multicolumn{2}{c}{$\stoptime(\cv)$} & \multicolumn{2}{c}{$\stoptime^{\rm w}(\cv)$} & $\stoptime^{\rm ww}(\cv)$ \\
\hline
Threshold $\mathbf{\cv}$ & Theory & Simulation & Theory & Simulation & Simulation  \\ 
\hline 
3.5 & 0.5227 & 0.5111 & 0.4889 & 0.4880 & 0.3046 \\   
3.6 & 0.4160 & 0.4079 & 0.3854 & 0.3868 & 0.2383 \\   
3.7 & 0.3207 & 0.3088 & 0.2943 & 0.2903 & 0.1777 \\   
3.8 & 0.2403 & 0.2311 & 0.2189 & 0.2143 & 0.1319 \\   
3.9 & 0.1756 & 0.1667 & 0.1588 & 0.1508 & 0.0923 \\   
4.0 & 0.1256 & 0.1180 & 0.1128 & 0.1066 & 0.0666 \\   
4.1 & 0.0881 & 0.0847 & 0.0787 & 0.0775 & 0.0499 \\   
4.2 & 0.0608 & 0.0593 & 0.0540 & 0.0541 & 0.0364 \\   
4.3 & 0.0413 & 0.0395 & 0.0365 & 0.0355 & 0.0249 \\   
4.4 & 0.0276 & 0.0265 & 0.0243 & 0.0233 & 0.0170 \\   
4.5 & 0.0183 & 0.0180 & 0.0160 & 0.0157 & 0.0118 \\
          \bottomrule
\end{tabu}
\caption*{FDR results based on simulated $1$-minute data with $\ell_n = 390$ (a day) for 10,000 replications. We choose the window size $w_n = 30$ minutes for $\stoptime^{\rm w}(\cv)$ and $(w_n, r_n) = (30, 5)$ minutes for $\stoptime^{\rm ww}(\cv)$. We set the jump truncation threshold $\zeta = 4\,\widehat{\sigma}_{t-1}^{\rm \, med}$.}
\label{table:FDR}
\end{table}

The above findings are corroborated by the FDR results in Table~\ref{table:FDR}, where we use a wider range of the threshold values, taking advantage of the fact that FDR can be assessed through simulations over deliberately chosen shorter time intervals. In particular, for each replication, we simulate one-minute prices for a day, and obtain the FDR as the fraction of replications with (at least) one false detection across our 5,000 replica. In general, the Monte Carlo FDRs are close to their theoretical values for $\stoptime(\cv)$ and $\stoptime^{\rm w}(\cv)$ by Corollary~\ref{cor:GLR_FDR}, implying accurate analytic approximations. The window limit $w_n$ does not impact the FDRs by much, while a moderate choice for $r_n$ may significantly reduce the false detection rate under the null. Finally, we note that this table can be used as a guide for choosing $\cv$, as well as $w_n$ and/or $r_n$, in empirical applications, perhaps following further simulations for alternative asset price dynamics.

\begin{figure}[!htb] 
\hspace*{-0mm}
\centering
\includegraphics[width=3.2in]{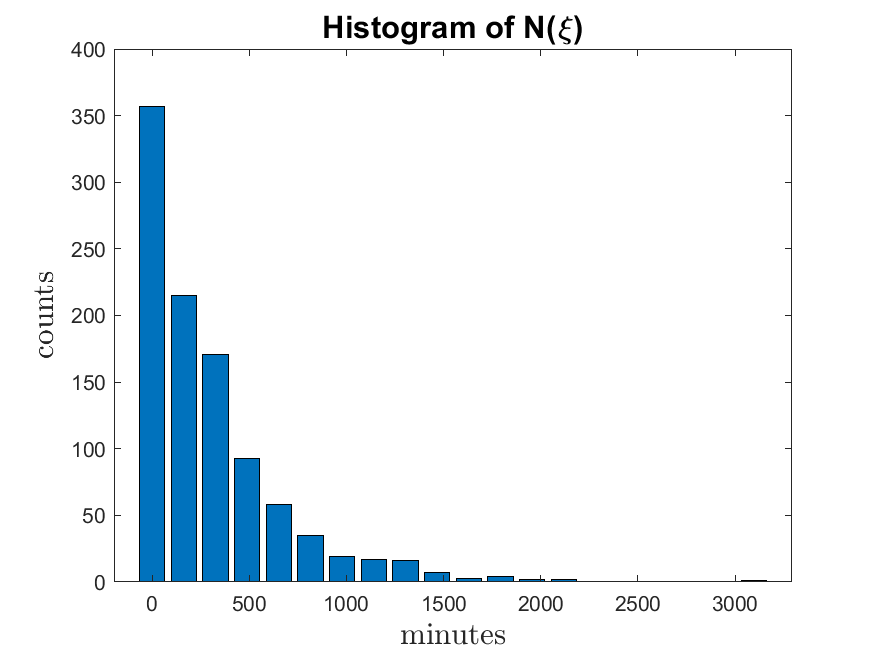} 
\includegraphics[width=3.2in]{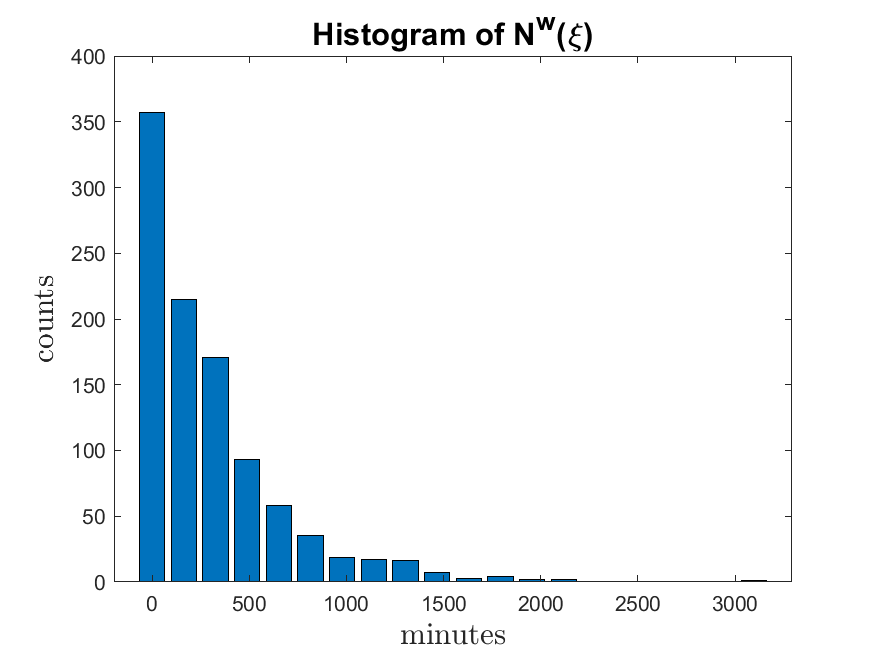} 
\caption{Histograms for $\stoptime(\cv)$ and $\stoptime^{\rm w}(\cv)$ based on simulated $1$-minute data and 1,000 replications. The window size is $w_n = 30$ minutes for $\stoptime^{\rm w}(\cv)$. The jump truncation parameter is $\zeta = 4\,\widehat{\sigma}_{t-1}^{\rm \, med}$ and the threshold is $\cv = 3.4$ for both cases.}
\label{fig:ARL}
\end{figure}

We round off our simulation study for the null measure $\prob_{\infty}$ with Figure~\ref{fig:ARL}. It displays histograms for the duration until (false) alarm associated with our detectors $\stoptime(\cv)$ ($w_n = \infty$) and $\stoptime^{\rm w}(\cv)$ ($w_n = 30$ minutes) based on 1,000 replications. The plots visually corroborate the approximate exponential distribution results in Theorem~\ref{thm:GLRw-CUSUM_null}.

We now turn to $\prob_{\tau_n}$ --- the probability measure with local It\^o semimartingale violation at time $\tau_n$ under the setting of Section~\ref{subsec:simu_setting}. Table~\ref{table:EDD} reports the (true) detection rates (denoted $\mathbf{r}$), i.e., what fraction of the 1,000 replications contain (at least) one detection, and the estimated \textit{expected detection delay} or EDD defined as,
\begin{align} \label{eqn:EDD}
\E_{\tau_n\,} [\, \stoptime^{\rm w}(\cv)-\tau_n \,|\, \stoptime^{\rm w}(\cv)>\tau_n \,] \, .
\end{align}
The EDD is regarded as the analogue of the ARL for the alternative measure $\prob_{\tau_n}$. In fact, ARL versus EDD is the canonical criteria in the sequential detection literature (see, e.g., \cite{lorden1971procedures}), as size versus power for hypothesis testing. Unfortunately, an explicit approximation for the EDD is not available in our case, and we can only derive BDD, which is a bound in probability on $\stoptime^{\rm w}(\cv)$. Reporting EDD in the simulation, however, makes it easier to assess the size of $\stoptime^{\rm w}(\cv)$ under the alternative.

\begin{table}[!htb] 
	\centering
	\small
	\caption{Detection Rate ($\mathbf{r}$) and Expected Detection Delay (EDD)}
\begin{tabu}{l|c c c|c c c}
          \hline\hline
 $\mathbf{r}$/{\bf EDD} & \multicolumn{3}{c}{$\stoptime^{\rm w}(\cv)$} & \multicolumn{3}{c}{$\stoptime^{\rm ww}(\cv)$} \\
          \hline
 ~ & 10s & 30s & 60s & 10s & 30s & 60s  \\
	      \midrule  
\multicolumn{3}{l}{{\bf \textit{DB gradual jumps}}}  \\
0.55 & 0.853/30.37 & 0.772/13.14 & 0.744/~8.12 
     & 0.870/28.48 & 0.703/13.90 & 0.514/11.75  \\ 
0.65 & 0.948/13.28 & 0.948/~8.77 & 0.932/~6.03 
     & 0.948/14.06 & 0.887/~9.36 & 0.756/~8.69  \\ 
0.75 & 1.000/~7.41 & 0.999/~5.83 & 0.989/~5.14
     & 0.996/~9.00 & 0.972/~8.09 & 0.977/~6.72  \\ 
\multicolumn{3}{l}{\bf {\textit{PN gradual jumps}}}  \\
0.45 & 0.943/14.24 & 0.819/12.24 & 0.645/~9.16  
     & 0.938/14.86 & 0.800/13.50 & 0.621/10.65  \\
0.35 & 0.969/10.86 & 0.913/~9.28 & 0.778/~7.61  
     & 0.966/11.46 & 0.909/10.42 & 0.773/~9.03  \\
0.25 & 0.997/~7.47 & 0.962/~7.52 & 0.977/~5.49
     & 0.997/~8.75 & 0.958/~8.98 & 0.977/~6.80  \\
          \bottomrule
\end{tabu}
\caption*{Detection rate ($\mathbf{r}$) and expected detection delay (EDD) results based on simulated $10$-, $30$- and $60$-seconds data for $1000$ days. We set $w_n = 30$ minutes for $\stoptime^{\rm w}(\cv)$ and $(w_n, r_n) = (30, 5)$ minutes for $\stoptime^{\rm ww}(\cv)$, with jump truncation threshold $\zeta = 4\,\widehat{\sigma}_{t-1}^{\, \rm med}$. \vspace{-0.5in}}
\label{table:EDD}
\end{table}

We consider both DB and PN gradual jump patterns as specified above. In both cases, we explore our detectors $\stoptime^{\rm w}(\cv)$ with $w_n = 30$ minutes and $\stoptime^{\rm ww}(\cv)$ with $(w_n, r_n) = (30, 5)$ minutes under three sampling frequencies: $10$-seconds, $30$-seconds and $60$-seconds.

Both detectors provide satisfactory performance, delivering high detection rates and short detection delays. Specifically, comparing vertically within each panel, we see that $\mathbf{r}$ increases and EDD shrinks, as the violations grow more severe. Likewise, comparing horizontally within each panel, we find the detectors improving with higher sampling frequency --- generating larger $\mathbf{r}$'s and shorter detection delays (in seconds, obtained by multiplying the entries accordingly with $10$, $30$ and $60$). 
Experiments featuring even higher sampling frequencies, in turn, confirm our (asymptotic) immediate detection result. However, implementation at very high frequencies within an actual market setting is impacted by microstructure noise, so it is arguably less practically applicable unless one employs a scheme to explicitly account for ultra high-frequency frictions.

\section{Empirical Application} \label{sec:empirical}

\subsection{Data}
We apply our detectors to one-minute S\&P 500 equity (SPX) index futures data covering January 2007 - December 2020. We only use data for the regular trading hours and eliminate days with reduced trading hours, producing a sample of $3,524$ days. Following our theoretical framework, we use the nonparametric option-based volatility estimator, SV, proposed by \citet{todorov2019nonparametric}. We rescale SV using the previous day's truncated volatility (TV) so that they are at the same level.

\subsection{Detection of semimartingale violations for S\&P 500}
We apply $\stoptime^{\rm ww}(\cv)$ with $(w_n,r_n) = (30,5)$ minutes and $\cv = 4$ for our one-minute data each day, implying we initiate the procedure $30$ minutes after the market open. Given this design, we lose power in terms of detecting violation periods shorter than $5$ minutes. From Table~\ref{table:FDR}, we see that this detector has a probability of about $6\%$ to trigger a false alarm at the daily level.

We further equip our detection with an \textit{ad hoc} identification rule to pin down It\^o semimartingale violation regions. Specifically, we define it as the union of all (probably overlapping) intervals, $[\,l_a, l_b\,]$, such that
\begin{align}
[\,l_a, l_b\,] = \mathop{\rm arg~max}\limits_{0 \leq l_a - w_n - r_n < l_b - r_n} \frac{|\widehat{\bm}_{l_a,l_b}|}{\sqrt{(l_b-l_a)\Delta_n}} > \cv.
\end{align}
That is, in words, we define the region as the union of the span of GLR-CUSUM statistics that exceed the chosen threshold $\cv$. 

Figures~\ref{fig:violation_eg1} and \ref{fig:violation_eg2} provide illustrations for two trading days, August 7, 2007 and August 30, 2019, with detected It\^o semimartingale violations. The price paths are red, $\stoptime^{\rm ww}(\cv)$ rejections are indicated by dark grey vertical lines, and the violation regions by light grey areas. In the bottom panel, we also provide $30$-minutes rolling window TV in blue and the option-based SV (rescaled by the same day's TV average) in orange.

On August 7, 2007, the price displays a gradual upward trend, until an abrupt 20 point crash over 2:15--2:30 pm, followed by a rapid recovery, taking the price beyond the original level. Our detector swiftly signals a violation as the flash-crash pattern emerges, triggering an alarm within a few minutes. As shown in the bottom plot, this dramatic price pattern induces an explosion in the local realized volatility measure, which is not consistent with our option-based SV measure, confirming the potential for dramatic biases in return-based volatility measurement. Similar findings hold for our second example on August 30, 2019, where a gradual jump initiates before 11:00 am. Again, our detector captures it expeditiously.

\begin{figure}[H] 
\hspace*{-0mm}
\centering
\includegraphics[width=6in]{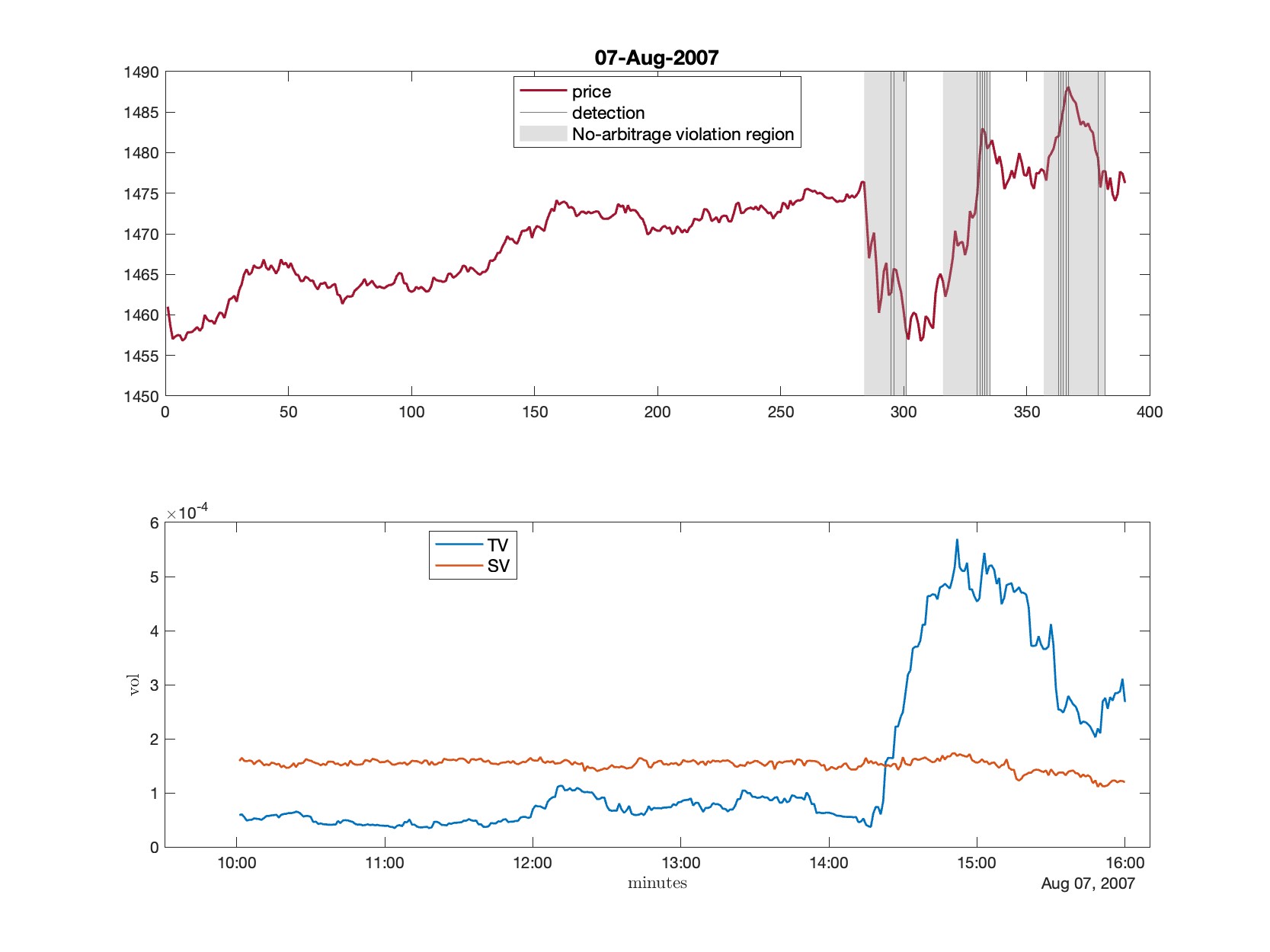}
\\ \vspace{-0.15in}
\caption{\small{{\bf Upper plot}: Price (red), Detection points (black vertical lines) and PN-regions (in grey) for August 7, 2007. {\bf Bottom plot}: same day's rolling window truncated realized volatility (TV, in blue) and spot volatility (SV, in orange) standardized by the day's average TV. Detection based on $\stoptime^{\rm ww}(\cv)$ with $(w_n,r_n) = (30,5)$ minutes, $\cv = 4$, and jump truncation threshold $\zeta = 4\,\widehat{\sigma}_{t-1}^{\, med}$.}}
\label{fig:violation_eg1}
\end{figure}

Table~\ref{table:violationregioncounts} reports the daily count of It\^o semimartingale violation regions detected by $\stoptime^{\rm ww}(\cv)$ with window length $(w_n,r_n) = (30,5)$ minutes under various values for the threshold $\cv$. We also split the trading day (of $390$ minutes) into three periods --- the ``Morning'' for the first $130$ minutes, the ``Noon'' for the middle $130$ minutes, and the ``Afternoon'' for the rest --- and report the associated count of violation regions initiated within that period. Finally, Figure~\ref{fig:histo_durations} plots the corresponding histograms for the duration of these violation regions for the scenario with $\cv = 4$.

\begin{figure}[!htb] 
\hspace*{-0mm}
\centering
\includegraphics[width=6in]
{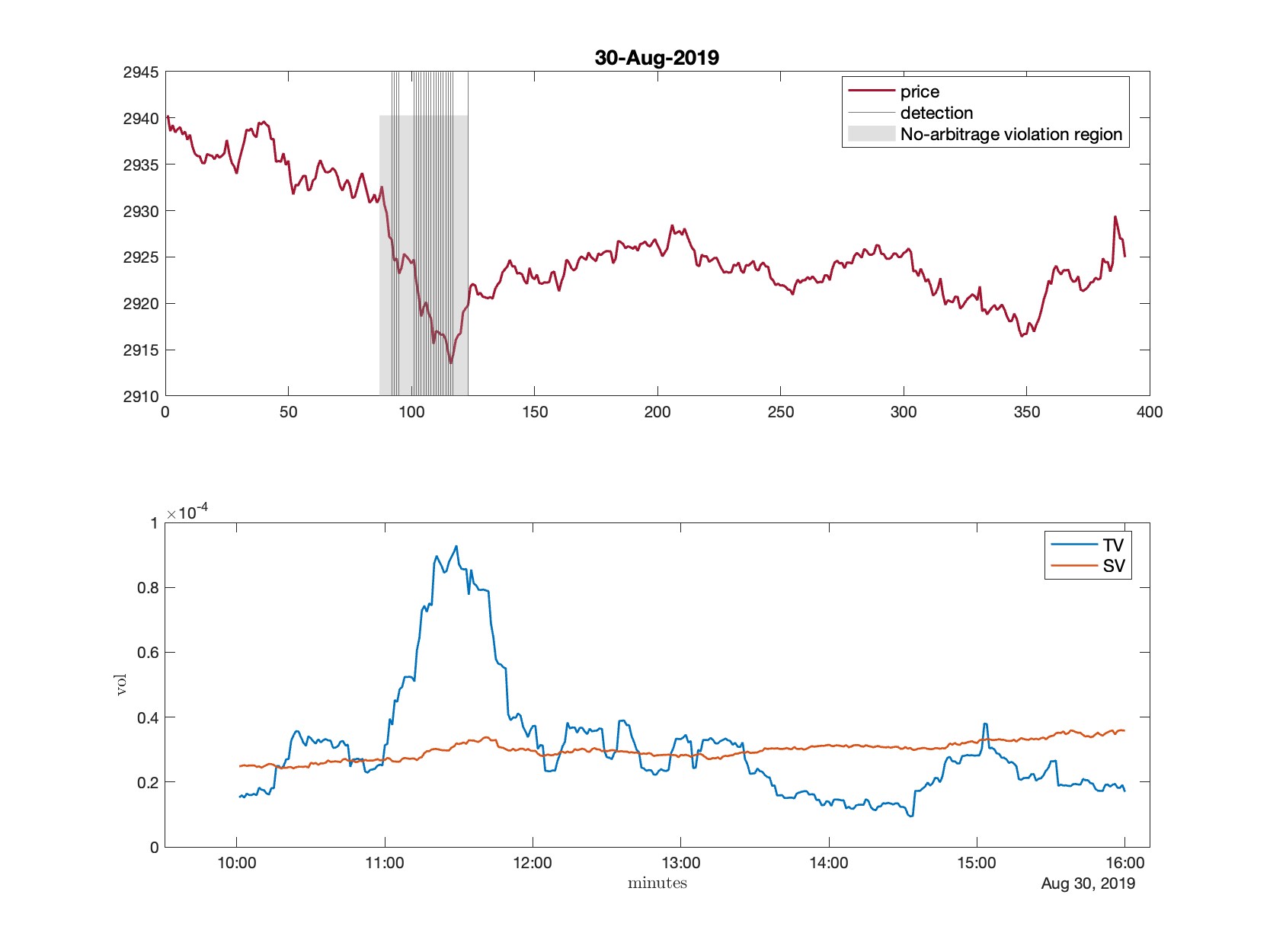}
\\ \vspace{-0.15in}
\caption{\small{{\bf Upper plot}: Price process (in red), Detection points (black vertical lines) and PN-regions (in grey) for August 30, 2019. {\bf Bottom plot}: same day's rolling window truncated volatility (TV, in blue) and spot volatility (SV, in orange) standardized by the day's average TV. Detection based on $\stoptime^{\rm ww}(\cv)$ with $(w_n,r_n) = (30,5)$ minutes, $\cv = 4$, and jump truncation threshold $\zeta = 4\,\widehat{\sigma}_{t-1}^{\, med}$.}}
\label{fig:violation_eg2}
\end{figure}

We find that ``Morning'' generates most violations, with the ``Noon'' containing only slightly fewer, while ``Afternoon'' produces the least, even if the difference is not striking. In sum, the violations are frequent and they are not particularly prone to occur in specific regions of the trading day. Finally, from the histograms in Figure~\ref{fig:histo_durations}, we note that the duration of a typical violation is short, with the clear majority lasting less than 20 minutes. Moreover, the violation regions starting in the morning tend to be somewhat shorter lived than those initiated during noon or afternoon.

\begin{table}[!htb] 
	\centering
	\caption{Count of Regions with It\^o Semimartingale Violations}
\begin{tabu}{c|c c c c}
          \hline\hline
Threshold $\cv$ & Total & Morning & Noon & Afternoon  \\
	      \midrule   
4 & 2186 & 824 & 738 & 624 \\ 
4.1 & 1901 & 714 & 641 & 546 \\ 
4.2 & 1638 & 620 & 542 & 476 \\ 
4.3 & 1398 & 520 & 461 & 417 \\ 
4.4 & 1172 & 438 & 378 & 356 \\ 
4.5 & 994 & 364 & 318 & 312 \\ 
          \bottomrule
\end{tabu}
\\ \vspace{-0.02in}
\caption*{The table reports the number of It\^o semimartingale violations and their distribution within the trading day. We split the day ($390$ min) into blocks: Morning (9:30 am--11:40 am), Noon (11:40 am--1:50 pm), and Afternoon (1:50 pm--4 pm), and count how many PN-regions are initiated within each block. For $\stoptime^{\rm ww}(\cv)$: window lengths $(w_n,r_n) = (30,5)$ minutes and jump truncation threshold $\zeta = 4\,\hat{\sigma}_{t-1}^{med}$.}
\label{table:violationregioncounts}
\end{table}

\begin{figure}[!htb] 
\hspace*{-0mm}
\centering
\includegraphics[width=3in]{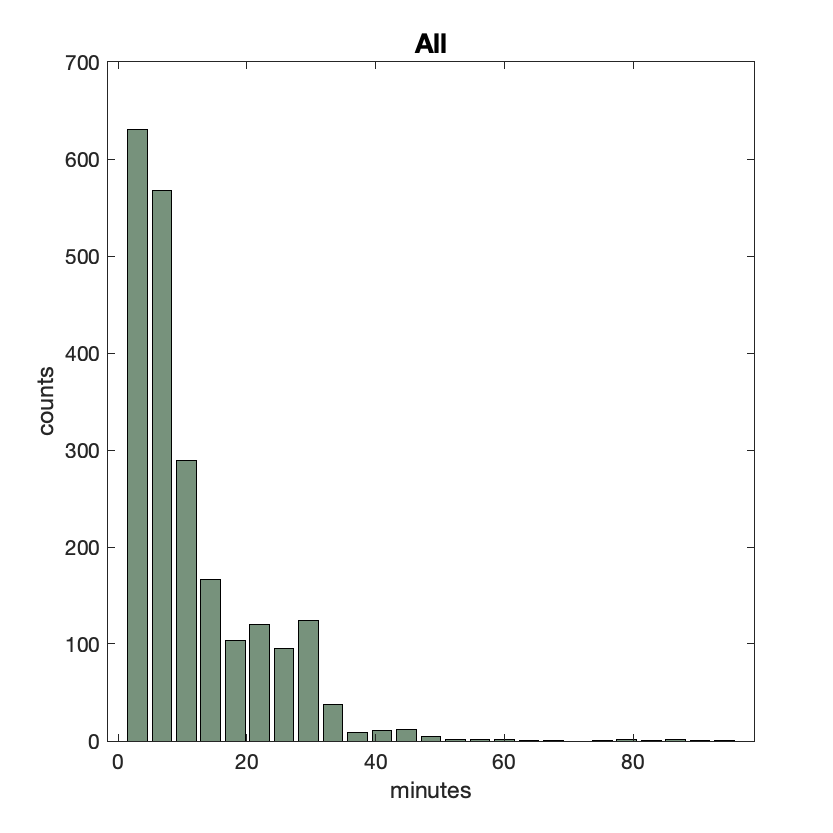} 
\includegraphics[width=3in]{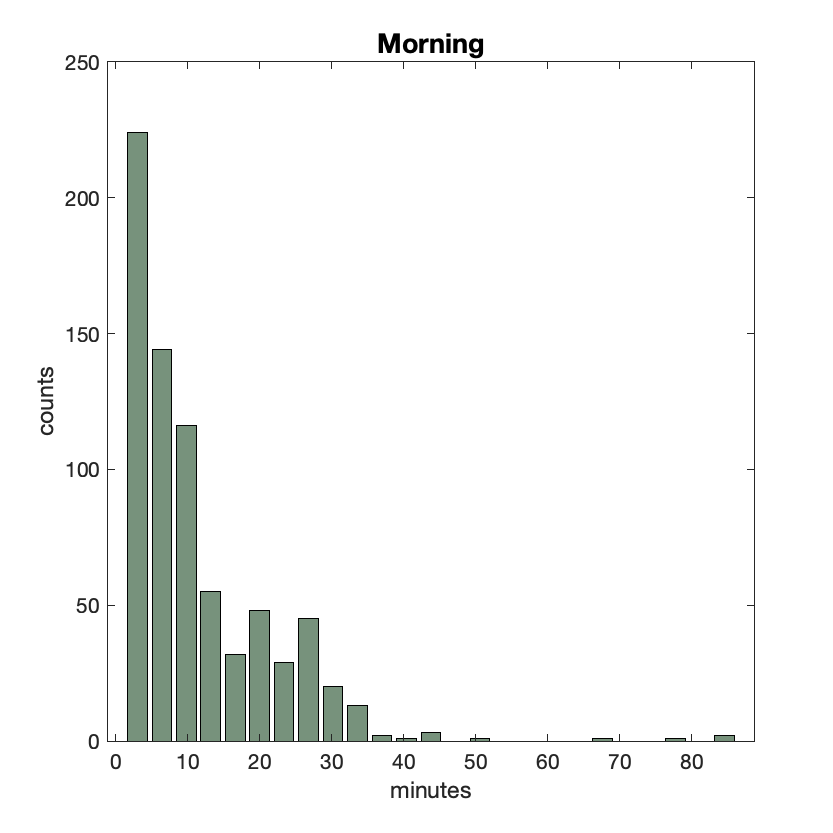} 
\includegraphics[width=3in]{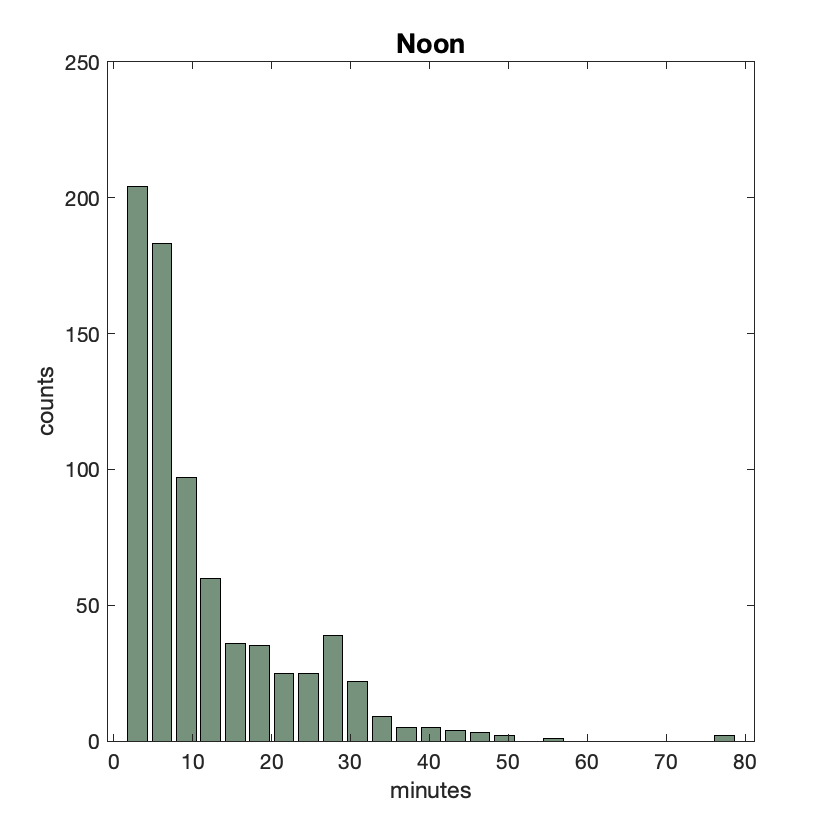} 
\includegraphics[width=3in]{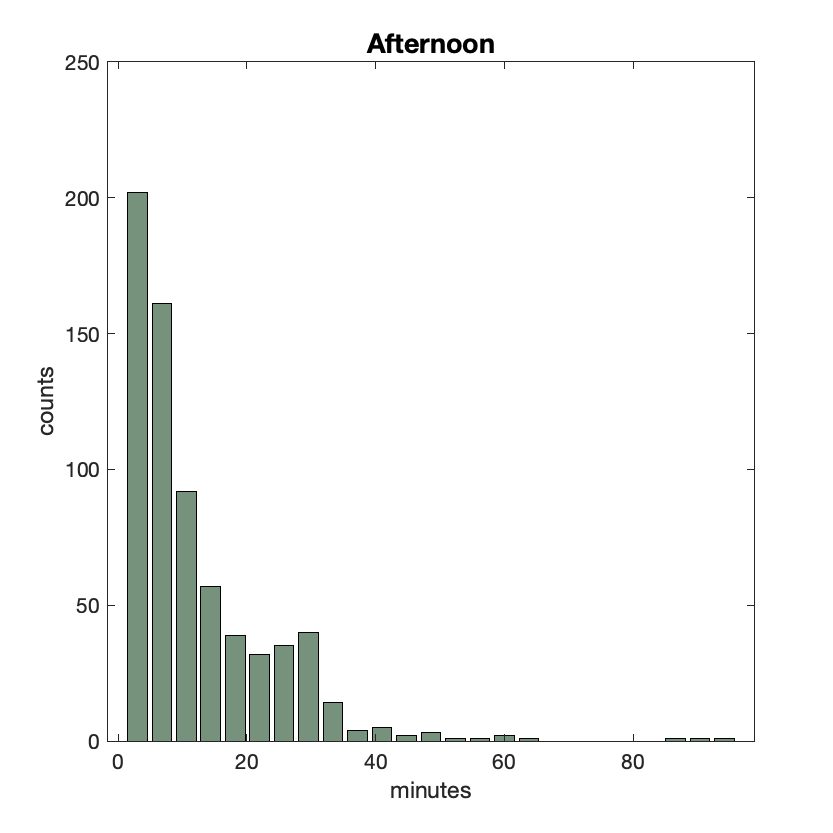} \\ \vspace{-0.2in}
\caption{Histograms of PN-region durations. Specifications: $\stoptime^{\rm ww}(\cv)$ with $(w_n,r_n) = (30,5)$ minutes and $\cv = 4$; Jump truncation threshold $\zeta = 4\,\widehat{\sigma}_{t-1}^{\, med}$.}
\label{fig:histo_durations}
\end{figure}

A manifestation of the It\^o semimartingale violation during the detected PN episodes is the divergence between return- and option-based spot volatility estimators. The bottom panels of Figures~\ref{fig:violation_eg1} and \ref{fig:violation_eg2} show that the divergences are large in these two specific cases. Table~\ref{table:vol_diff} reports the sample mean of the return- and option-based volatility estimates $TV$ and $SV$ over the hour before and after the initiation of a PN episode. The results reveal that the two volatility proxies are close before the PN episode but feature a substantial gap over the following hour, with the average $TV$ about $35\%$ higher than the average $SV$. This gap only grows larger, if we exclude days with FOMC announcements from the calculation, so the discrepancy is not driven by this type of macroeconomic announcements. These findings demonstrate how local deviations from the semimartingale assumption can distort the measurement of volatility.



\begin{table}[!htb] 
	\centering
	\caption{Volatility Estimates during Regions with It\^o Semimartingale Violations}
\begin{tabu}{c|c c c c}
          \hline\hline
~ & \multicolumn{2}{c}{1 hour before} & \multicolumn{2}{c}{1 hour after} \\
          \hline
~ & TV & SV & TV & SV \\
	      \midrule  
ALL     & 10.93$\times 10^{-5}$ & 10.58$\times 10^{-5}$ & 12.61$\times 10^{-5}$ & 10.53$\times 10^{-5}$ \\
No FOMC & 11.72$\times 10^{-5}$ & 11.56$\times 10^{-5}$ & 13.72$\times 10^{-5}$ & 11.51$\times 10^{-5}$ \\
          \hline
          \bottomrule
\end{tabu}
\caption*{This table reports the average volatility measures $TV$ and $SV$, one hour before and one hour after the starting point of PN episodes.}
\label{table:vol_diff}
\end{table}

Finally, Table \ref{table:yr_count} reports the number of PN episodes for each year across the sample given different threshold levels for detection. The PN events appear near uniformly distributed over the sample, with no particular clustering apparent during periods characterized by generally volatile or tranquil market conditions. Overall, we conclude that PN episodes are largely idiosyncratic events with no particular tendency to occur during specific times within the trading day or during specific market conditions.

\begin{table}[!htb] 
	\centering
	\caption{PN-episodes Yearly Counts}
\begin{tabu}{r|rrrrrr}
\hline\hline
Year & $\cv = 4$ & $\cv = 4.1$ & $\cv = 4.2$ & $\cv = 4.3$ & $\cv = 4.4$ & $\cv = 4.5$ \\ 
\midrule
2007 & 171 & 155 & 132 & 105 & 83 & 71 \\ 
2008 & 132 & 117 & 98 & 88 & 77 & 63 \\ 
2009 & 89 & 80 & 65 & 56 & 42 & 30 \\ 
2010 & 129 & 108 & 94 & 80 & 65 & 53 \\ 
2011 & 133 & 117 & 100 & 88 & 78 & 57 \\ 
2012 & 133 & 109 & 92 & 83 & 69 & 56 \\ 
2013 & 194 & 168 & 142 & 118 & 103 & 92 \\ 
2014 & 241 & 206 & 165 & 143 & 112 & 96 \\ 
2015 & 139 & 121 & 116 & 97 & 87 & 76 \\ 
2016 & 153 & 137 & 118 & 100 & 87 & 75 \\ 
2017 & 170 & 151 & 132 & 113 & 99 & 84 \\ 
2018 & 150 & 131 & 117 & 98 & 80 & 71 \\ 
2019 & 172 & 153 & 133 & 111 & 91 & 83 \\ 
2020 & 180 & 148 & 134 & 118 & 99 & 87 \\ 
\hline
\bottomrule
\end{tabu}
\caption*{This table reports annual persistent-noise episode counts from 2007 to 2020 under different threshold values.}
\label{table:yr_count}
\end{table}

\section{Conclusion} \label{sec:conclusion}

In this paper, we focus on real-time detection of local It\^o semimartingale violations that have drawn increasing attention in the recent literature. We propose CUSUM-type detectors as stopping rules exploiting high-frequency data. We show that they possess desirable theoretical properties under infill asymptotics. Specifically, for a suitably chosen average run length (the average sample length until a false alarm), our detectors enable ``quick'' detection of a violation, once it occurs. Our formal interpretation of rapid detection is that the bound on the detection delay (BDD) shrinks to zero under infill asymptotics, as the sampling interval $\Delta_n$ goes to zero. These properties are corroborated through simulations calibrated to reflect key features of market data. Finally, we apply our detector to S\&P 500 equity (SPX) index futures data. We obtain timely detection of a nontrivial number of short-lived episodes involving likely semimartingale violations. Such turbulent market events are critical for real-time decision making. For example, they may indicate temporary market malfunction motivating exchange or regulatory action, they may induce the termination of trading strategies and algorithms among active market participants triggering a period of fleeting liquidity, and they signal problems in extracting reliable high-frequency based market volatility and risk measures.

\clearpage
\bibliographystyle{asa}
\bibliography{reference}

\begin{thebibliography}{56}
\newcommand{\enquote}[1]{``#1''}
\expandafter\ifx\csname natexlab\endcsname\relax\def\natexlab#1{#1}\fi

\bibitem[{Andersen et~al.(2022)Andersen, Archakov, Cebiroglu, and
  Hautsch}]{AACH}
Andersen, T., Archakov, I., Cebiroglu, G., and Hautsch, N. (2022),
  \enquote{Local mispricing and market microstructure noise: A parametric
  perspective,} \textit{Journal of Econometrics}, 230, 510--534.

\bibitem[{Andersen et~al.(2021)Andersen, Li, Todorov, and
  Zhou}]{andersen2021volatility}
Andersen, T.~G., Li, Y., Todorov, V., and Zhou, B. (2021), \enquote{Volatility
  measurement with pockets of extreme return persistence,} \textit{Journal of
  Econometrics}.

\bibitem[{Andreou and Ghysels(2006)}]{andreou2006monitoring}
Andreou, E. and Ghysels, E. (2006), \enquote{Monitoring disruptions in
  financial markets,} \textit{Journal of Econometrics}, 135, 77--124.

\bibitem[{Andreou and Ghysels(2008)}]{andreou2008quality}
--- (2008), \enquote{Quality control for structural credit risk models,}
  \textit{Journal of Econometrics}, 146, 364--375.

\bibitem[{Andreou et~al.(2013)Andreou, Ghysels, and
  Kourtellos}]{andreou2013should}
Andreou, E., Ghysels, E., and Kourtellos, A. (2013), \enquote{Should
  macroeconomic forecasters use daily financial data and how?} \textit{Journal
  of Business \& Economic Statistics}, 31, 240--251.

\bibitem[{Andrews(1993)}]{andrews1993tests}
Andrews, D.~W. (1993), \enquote{Tests for parameter instability and structural
  change with unknown change point,} \textit{Econometrica: Journal of the
  Econometric Society}, 821--856.

\bibitem[{Aue and Horv{\'a}th(2004)}]{aue2004delay}
Aue, A. and Horv{\'a}th, L. (2004), \enquote{Delay time in sequential detection
  of change,} \textit{Statistics \& Probability Letters}, 67, 221--231.

\bibitem[{Aue et~al.(2006)Aue, Horv{\'a}th, Hu{\v{s}}kov{\'a}, and
  Kokoszka}]{aue2006change}
Aue, A., Horv{\'a}th, L., Hu{\v{s}}kov{\'a}, M., and Kokoszka, P. (2006),
  \enquote{Change-point monitoring in linear models,} \textit{The Econometrics
  Journal}, 9, 373--403.

\bibitem[{Bai(1996)}]{bai1996testing}
Bai, J. (1996), \enquote{Testing for parameter constancy in linear regressions:
  an empirical distribution function approach,} \textit{Econometrica: Journal
  of the Econometric Society}, 597--622.

\bibitem[{Bai and Perron(1998)}]{bai1998estimating}
Bai, J. and Perron, P. (1998), \enquote{Estimating and testing linear models
  with multiple structural changes,} \textit{Econometrica}, 47--78.

\bibitem[{Ba{\'n}bura et~al.(2013)Ba{\'n}bura, Giannone, Modugno, and
  Reichlin}]{banbura2013now}
Ba{\'n}bura, M., Giannone, D., Modugno, M., and Reichlin, L. (2013),
  \enquote{Now-casting and the real-time data flow,} in \textit{Handbook of
  economic forecasting}, Elsevier, vol.~2, pp. 195--237.

\bibitem[{Barndorff-Nielsen et~al.(2009)Barndorff-Nielsen, Hansen, Lunde, and
  Shephard}]{barndorff2009realized}
Barndorff-Nielsen, O.~E., Hansen, P.~R., Lunde, A., and Shephard, N. (2009),
  \enquote{Realized kernels in practice: Trades and quotes,} \textit{The
  Econometrics Journal}, 12, C1--C32.

\bibitem[{B{\"u}cher et~al.(2017)B{\"u}cher, Hoffmann, Vetter, Dette,
  et~al.}]{bucher2017nonparametric}
B{\"u}cher, A., Hoffmann, M., Vetter, M., Dette, H., et~al. (2017),
  \enquote{Nonparametric tests for detecting breaks in the jump behaviour of a
  time-continuous process,} \textit{Bernoulli}, 23, 1335--1364.

\bibitem[{Chan and Lai(2003)}]{chan2003saddlepoint}
Chan, H.~P. and Lai, T.~L. (2003), \enquote{Saddlepoint approximations and
  nonlinear boundary crossing probabilities of Markov random walks,}
  \textit{The Annals of Applied Probability}, 13, 395--429.

\bibitem[{Christensen et~al.(2022)Christensen, Oomen, and
  Ren{\`o}}]{christensen2020drift}
Christensen, K., Oomen, R., and Ren{\`o}, R. (2022), \enquote{The drift burst
  hypothesis,} \textit{Journal of Econometrics}, 227, 461--497.

\bibitem[{Christensen et~al.(2014)Christensen, Oomen, and
  Podolskij}]{christensen2014fact}
Christensen, K., Oomen, R.~C., and Podolskij, M. (2014), \enquote{Fact or
  friction: Jumps at ultra high frequency,} \textit{Journal of Financial
  Economics}, 114, 576--599.

\bibitem[{Chu et~al.(1996)Chu, Stinchcombe, and White}]{chu1996monitoring}
Chu, C.-S.~J., Stinchcombe, M., and White, H. (1996), \enquote{Monitoring
  structural change,} \textit{Econometrica: Journal of the Econometric
  Society}, 1045--1065.

\bibitem[{Delbaen and Schachermayer(1994)}]{delbaen1994general}
Delbaen, F. and Schachermayer, W. (1994), \enquote{A general version of the
  fundamental theorem of asset pricing,} \textit{Mathematische annalen}, 300,
  463--520.

\bibitem[{Diba and Grossman(1988)}]{diba1988explosive}
Diba, B.~T. and Grossman, H.~I. (1988), \enquote{Explosive rational bubbles in
  stock prices?} \textit{The American Economic Review}, 78, 520--530.

\bibitem[{Dragalin(1994)}]{dragalin1994optimality}
Dragalin, V. (1994), \enquote{Optimality of generalized Cusum procedure in
  quickest detection problem,} \textit{Proceedings of the Steklov Institute of
  Mathematics-AMS Translation}, 202, 107--120.

\bibitem[{Elliott and M{\"u}ller(2014)}]{elliott2014pre}
Elliott, G. and M{\"u}ller, U.~K. (2014), \enquote{Pre and post break parameter
  inference,} \textit{Journal of Econometrics}, 180, 141--157.

\bibitem[{Embrechts et~al.(2013)Embrechts, Kl{\"u}ppelberg, and
  Mikosch}]{embrechts2013modelling}
Embrechts, P., Kl{\"u}ppelberg, C., and Mikosch, T. (2013), \textit{Modelling
  extremal events: for insurance and finance}, vol.~33, Springer Science \&
  Business Media.

\bibitem[{Evans(2005)}]{evans2005we}
Evans, M. (2005), \enquote{Where are we now? Real-time estimates of the
  macroeconomy,} \textit{International Journal of Central Banking}, 1,
  127--175.

\bibitem[{Figueroa-L{\'o}pez and {\'O}lafsson(2019)}]{figueroa2019change}
Figueroa-L{\'o}pez, J.~E. and {\'O}lafsson, S. (2019), \enquote{Change-point
  detection for L{\'e}vy processes,} \textit{Annals of Applied Probability},
  29, 717--738.

\bibitem[{Giannone et~al.(2008)Giannone, Reichlin, and
  Small}]{giannone2008nowcasting}
Giannone, D., Reichlin, L., and Small, D. (2008), \enquote{Nowcasting: The
  real-time informational content of macroeconomic data,} \textit{Journal of
  monetary economics}, 55, 665--676.

\bibitem[{Hagwood and Woodroofe(1982)}]{hagwood1982expansion}
Hagwood, C. and Woodroofe, M. (1982), \enquote{On the expansion for expected
  sample size in non-linear renewal theory,} \textit{The Annals of
  Probability}, 844--848.

\bibitem[{Horv{\'a}th et~al.(2004)Horv{\'a}th, Hu{\v{s}}kov{\'a}, Kokoszka, and
  Steinebach}]{horvath2004monitoring}
Horv{\'a}th, L., Hu{\v{s}}kov{\'a}, M., Kokoszka, P., and Steinebach, J.
  (2004), \enquote{Monitoring changes in linear models,} \textit{Journal of
  statistical Planning and Inference}, 126, 225--251.

\bibitem[{Horvath et~al.(2007)Horvath, Kokoszka, and
  Steinebach}]{horvath2007sequential}
Horvath, L., Kokoszka, P., and Steinebach, J. (2007), \enquote{On sequential
  detection of parameter changes in linear regression,} \textit{Statistics \&
  probability letters}, 77, 885--895.

\bibitem[{Kirilenko et~al.(2017)Kirilenko, Kyle, Samadi, and
  Tuzun}]{kirilenko2017flash}
Kirilenko, A., Kyle, A.~S., Samadi, M., and Tuzun, T. (2017), \enquote{The
  flash crash: High-frequency trading in an electronic market,} \textit{The
  Journal of Finance}, 72, 967--998.

\bibitem[{Lai and Siegmund(1979)}]{lai1979nonlinear}
Lai, T. and Siegmund, D. (1979), \enquote{A nonlinear renewal theory with
  applications to sequential analysis II,} \textit{The Annals of Statistics},
  60--76.

\bibitem[{Lai(1995)}]{lai1995sequential}
Lai, T.~L. (1995), \enquote{Sequential changepoint detection in quality control
  and dynamical systems,} \textit{Journal of the Royal Statistical Society:
  Series B (Methodological)}, 57, 613--644.

\bibitem[{Lai and Shan(1999)}]{lai1999efficient}
Lai, T.~L. and Shan, J.~Z. (1999), \enquote{Efficient recursive algorithms for
  detection of abrupt changes in signals and control systems,} \textit{IEEE
  Transactions on Automatic Control}, 44, 952--966.

\bibitem[{Leisch et~al.(2000)Leisch, Hornik, and Kuan}]{leisch2000monitoring}
Leisch, F., Hornik, K., and Kuan, C.-M. (2000), \enquote{Monitoring structural
  changes with the generalized fluctuation test,} \textit{Econometric Theory},
  16, 835--854.

\bibitem[{Lorden(1971)}]{lorden1971procedures}
Lorden, G. (1971), \enquote{Procedures for reacting to a change in
  distribution,} \textit{The Annals of Mathematical Statistics}, 42,
  1897--1908.

\bibitem[{Menkveld and Yueshen(2019)}]{menkveld2019flash}
Menkveld, A.~J. and Yueshen, B.~Z. (2019), \enquote{The Flash Crash: A
  cautionary tale about highly fragmented markets,} \textit{Management
  Science}, 65, 4470--4488.

\bibitem[{Moustakides et~al.(1986)}]{moustakides1986optimal}
Moustakides, G.~V. et~al. (1986), \enquote{Optimal stopping times for detecting
  changes in distributions,} \textit{Annals of Statistics}, 14, 1379--1387.

\bibitem[{Moustakides et~al.(2004)}]{moustakides2004optimality}
--- (2004), \enquote{Optimality of the CUSUM procedure in continuous time,}
  \textit{The Annals of Statistics}, 32, 302--315.

\bibitem[{Page(1954)}]{page1954continuous}
Page, E.~S. (1954), \enquote{Continuous inspection schemes,}
  \textit{Biometrika}, 41, 100--115.

\bibitem[{Phillips et~al.(2014)Phillips, Shi, and
  Yu}]{phillips2014specification}
Phillips, P.~C., Shi, S., and Yu, J. (2014), \enquote{Specification sensitivity
  in right-tailed unit root testing for explosive behaviour,} \textit{Oxford
  Bulletin of Economics and Statistics}, 76, 315--333.

\bibitem[{Phillips et~al.(2015{\natexlab{a}})Phillips, Shi, and
  Yu}]{phillips2015testinga}
--- (2015{\natexlab{a}}), \enquote{Testing for multiple bubbles: Historical
  episodes of exuberance and collapse in the S\&P 500,} \textit{International
  economic review}, 56, 1043--1078.

\bibitem[{Phillips et~al.(2015{\natexlab{b}})Phillips, Shi, and
  Yu}]{phillips2015testingb}
--- (2015{\natexlab{b}}), \enquote{Testing for multiple bubbles: Limit theory
  of real-time detectors,} \textit{International Economic Review}, 56,
  1079--1134.

\bibitem[{Phillips and Shi(2018)}]{phillips2018financial}
Phillips, P.~C. and Shi, S.-P. (2018), \enquote{Financial bubble implosion and
  reverse regression,} \textit{Econometric Theory}, 34, 705--753.

\bibitem[{Phillips et~al.(2011)Phillips, Wu, and Yu}]{phillips2011explosive}
Phillips, P.~C., Wu, Y., and Yu, J. (2011), \enquote{Explosive behavior in the
  1990s Nasdaq: When did exuberance escalate asset values?}
  \textit{International economic review}, 52, 201--226.

\bibitem[{Phillips and Yu(2011)}]{phillips2011dating}
Phillips, P.~C. and Yu, J. (2011), \enquote{Dating the timeline of financial
  bubbles during the subprime crisis,} \textit{Quantitative Economics}, 2,
  455--491.

\bibitem[{Pollak and Siegmund(1975)}]{pollak1975approximations}
Pollak, M. and Siegmund, D. (1975), \enquote{Approximations to the expected
  sample size of certain sequential tests,} \textit{The Annals of Statistics},
  1267--1282.

\bibitem[{Ritov(1990)}]{ritov1990decision}
Ritov, Y. (1990), \enquote{Decision theoretic optimality of the CUSUM
  procedure,} \textit{The Annals of Statistics}, 1464--1469.

\bibitem[{Roberts(1966)}]{roberts1966comparison}
Roberts, S. (1966), \enquote{A comparison of some control chart procedures,}
  \textit{Technometrics}, 8, 411--430.

\bibitem[{Ross(1976)}]{ross1976arbitrage}
Ross, S.~A. (1976), \enquote{The arbitrage theory of capital asset pricing,}
  \textit{Journal of Economic Theory}, 13, 341--360.

\bibitem[{Shiryaev(1963)}]{shiryaev1963optimum}
Shiryaev, A.~N. (1963), \enquote{On optimum methods in quickest detection
  problems,} \textit{Theory of Probability \& Its Applications}, 8, 22--46.

\bibitem[{Shiryaev(1996)}]{shiryaev1996minimax}
--- (1996), \enquote{Minimax optimality of the method of cumulative sums
  (CUSUM) in the case of continuous time,} \textit{Russian Mathematical
  Surveys}, 51, 750.

\bibitem[{Siegmund(1988)}]{siegmund1988approximate}
Siegmund, D. (1988), \enquote{Approximate tail probabilities for the maxima of
  some random fields,} \textit{The Annals of Probability}, 487--501.

\bibitem[{Siegmund and Venkatraman(1995)}]{siegmund1995using}
Siegmund, D. and Venkatraman, E. (1995), \enquote{Using the generalized
  likelihood ratio statistic for sequential detection of a change-point,}
  \textit{The Annals of Statistics}, 255--271.

\bibitem[{Srivastava and Wu(1993)}]{srivastava1993comparison}
Srivastava, M. and Wu, Y. (1993), \enquote{Comparison of EWMA, CUSUM and
  Shiryayev-Roberts procedures for detecting a shift in the mean,} \textit{The
  Annals of Statistics}, 21, 645--670.

\bibitem[{Todorov(2019)}]{todorov2019nonparametric}
Todorov, V. (2019), \enquote{Nonparametric spot volatility from options,}
  \textit{The Annals of Applied Probability}, 29, 3590--3636.

\bibitem[{Xie and Siegmund(2013)}]{xie2013sequential}
Xie, Y. and Siegmund, D. (2013), \enquote{Sequential multi-sensor change-point
  detection,} in \textit{2013 Information Theory and Applications Workshop
  (ITA)}, IEEE, pp. 1--20.

\bibitem[{Yao(1993)}]{yao1993boundary}
Yao, Q. (1993), \enquote{Boundary-crossing probabilities of some random fields
  related to likelihood ratio tests for epidemic alternatives,} \textit{Journal
  of applied probability}, 30, 52--65.

\end{thebibliography}

\clearpage
\appendix
\section{Proofs} \label{appsec:proofs}

\begin{proof}[Proof of Proposition~\ref{prop:consistency_BM}]
As it will become clear from the proof below, it suffices to conduct the analysis for the case when $0\leq l_1^n\leq k_n$ and $0\leq l_2^n\leq k_n$. To simplify the notation, we will use $l_1$ and $l_2$ for $l_1^n$ and $l_2^n$, respectively, in the proof. We can write
\begin{equation}
\begin{split}
&\mathbb{P}_{0}\left(\max_{0<l_1<l_2\leq k_n} \frac{|\widehat{\bm}_{l_1,l_2} - \widecheck{\bm}_{l_1,l_2}|}{\sqrt{(l_2-l_1)\Delta_n}}>\xi\right) \\ &~~~~~~~~ \leq \sum_{i=1}^{k_n}\mathbb{P}_{0}\left(\frac{1}{\sqrt{\Delta_n}}\left|\frac{\Delta_i^nX}{\hat{\sigma}_{(i-1)\Delta_n}}\indicator_{\{|\Delta_i^nX|<\zeta\Delta_n^\varpi\}} - \Delta_i^nW\right|>\frac{\xi}{\sqrt{k_n}}\right),
\end{split}
\end{equation}
and further
\begin{equation}
\begin{split}
&\mathbb{P}_{0}\left(\frac{1}{\sqrt{\Delta_n}}\left|\frac{\Delta_i^nX}{\hat{\sigma}_{(i-1)\Delta_n}}\indicator_{\{|\Delta_i^nX|<\zeta\Delta_n^\varpi\}} - \Delta_i^nW\right|>\frac{\xi}{\sqrt{k_n}}\right)\\&~\leq \mathbb{P}_{0}\left(\frac{|\Delta_i^nX|}{\sqrt{\Delta_n}}\frac{|\hat{\sigma}_{(i-1)\Delta_n}-\sigma_{(i-1)\Delta_n}|}{\sigma_{(i-1)\Delta_n}}>\frac{C_0}{l_n}\frac{\xi}{4\sqrt{k_n}}\right)\\&~~~+\mathbb{P}_0\bigg(\frac{|\Delta_i^nJ|+|\int_{(i-1)\Delta_n}^{i\Delta_n}b_sds|}{\sqrt{\Delta_n}\sigma_{(i-1)\Delta_n}}>\frac{\xi}{4\sqrt{k_n}}\bigg) + \mathbb{P}_0\left(|\Delta_i^nX|\geq \zeta\Delta_n^\varpi\right) \\&~~~+\mathbb{P}_0\bigg(\frac{|\int_{(i-1)\Delta_n}^{i\Delta_n}(\sigma_s-\sigma_{(i-1)\Delta_n})dW_s|}{\sqrt{\Delta_n}\sigma_{(i-1)\Delta_n}}>\frac{\xi}{4\sqrt{k_n}}\bigg).
\end{split}
\end{equation}
Using Assumption 1 for the smoothness in expectation of $\sigma$, the integrability conditions of Assumption 1, Assumption 2 and applying Cauchy–Schwarz inequality, we have 
\begin{equation}
\mathbb{P}_{0}\left(\frac{|\Delta_i^nX|}{\sqrt{\Delta_n}}\frac{|\hat{\sigma}_{(i-1)\Delta_n}-\sigma_{(i-1)\Delta_n}|}{\sigma_{(i-1)\Delta_n}}>\frac{C_0}{l_n}\frac{\xi}{4\sqrt{k_n}}\right)\leq C_0\left(k_n\Delta_n + \frac{\sqrt{k_n}l_n\delta_n}{\xi}\right).
\end{equation}
Using the integrability conditions for the process $b$ and the jump size function  as well as the smoothness in expectation condition for $\sigma$ in Assumption 1, we have 
\begin{equation}
\mathbb{P}_0\left(\frac{|\Delta_i^nJ|+|\int_{(i-1)\Delta_n}^{i\Delta_n}b_sds|}{\sqrt{\Delta_n}\sigma_{(i-1)\Delta_n}}>\frac{\xi}{4\sqrt{k_n}}\right)\leq C_0\left(\frac{\sqrt{k_n\Delta_n}}{\xi}+k_n\Delta_n\right).
\end{equation}
Next, using the smoothness in expectation condition for $\sigma$ in Assumption 1, we have
\begin{equation}
\mathbb{P}_0\left(\frac{|\int_{(i-1)\Delta_n}^{i\Delta_n}(\sigma_s-\sigma_{(i-1)\Delta_n})dW_s|}{\sqrt{\Delta_n}\sigma_{(i-1)\Delta_n}}>\frac{\xi}{4\sqrt{k_n}}\right)\leq C_0\frac{k_n\Delta_n}{1\wedge \xi^2}.
\end{equation}
Finally, given the integrability conditions in Assumption 1, we have
\begin{equation}
\mathbb{P}_0\left(|\Delta_i^nX|\geq \zeta\Delta_n^\varpi\right)\leq C_0\Delta_n^{1-\varpi}.
\end{equation}
Combining the above bounds, we get the result to be proved.
\end{proof}

\bigskip

\begin{proof}[Proof of Theorem~\ref{thm:GLRw-CUSUM_null}]
Recalling that $\widecheck{\bm}$ is defined as a discrete random walk with $N(0,\Delta_n)$ increments, by \cite[Proposition 1]{siegmund1995using}, we have 
\begin{align} \label{eqn:siegmundProp1}
\measure\left\{ \max_{0<l_1<l_2<m} \frac{|\widecheck{\bm}_{l_1,l_2}|}{\sqrt{(l_2-l_1)\Delta_n}} > \cv \right\} 
\sim 
m\ARLHc\cv\phi(\cv),
\end{align}
as $\cv\to\infty$, for $m = \bar{c}\cv^2$ with some $\bar{c}>0$ fixed, where 
\begin{align*}
\ARLHc := \int_{\bar{c}^{-1/2}}^{\infty}x\nu^2(x)\dd x - \bar{c}^{-1}\int_{\bar{c}^{-1/2}}^{\infty}x^{-1}\nu^2(x)\dd x. 
\end{align*}

The above result based on the analysis of the random fields (see \citet{siegmund1988approximate}) shows that the probability of the maximum of random fields $|\widecheck{\bm}_{l_1,l_2}|/\sqrt{l_2-l_1}$ for $0 < l_1 < l_2 \leq m$ exceeding $\cv$, when $m$ is of order $\cv^2$, goes to $0$ exponentially fast.

From the statement of Theorem~\ref{thm:GLRw-CUSUM_null}, the expected sample size should be of order $1/(\cv\phi(\cv)) \sim \exp(\cv^2)/\cv$, and it is much bigger than the order of $m$ (which is $\cv^2$). We follow the trick in \citet{siegmund1995using} to complete the proof of the theorem. Let the window size $w_n$ be chosen such that $w_n \geq m$ and $w_n/m$ is a positive integer. More specifically, we set $m = w_n$ if $a$ in the statement of the theorem is finite and if $a=\infty$ we set $m = w_n/\lfloor w_n/(\overline{c}\cv^2)\rfloor$. We note that with the latter choice, $m/\cv^2\rightarrow \overline{c}$. Picking $M$ with $M = r\E_{\infty}[\stoptime^{\rm w}(\cv)] \sim r/(\ARLH_a\cv\phi(\cv))$ for some fixed $r$ bounded away from 0, we have
\begin{align} \label{eqn:GLRw_hattocheck}
\measure\left\{\stoptime^{\rm w} > M\right\} 
&\leq \measure\left\{ \max_{0 \leq j<\frac{M}{m}}\max_{jm<l_1<l_2<(j+1)m} \frac{|\widehat{\bm}_{l_1,l_2}|}{\sqrt{(l_2-l_1)\Delta_n}} < \cv \right\} \nonumber \\
&\leq \measure\left\{ \max_{0 \leq j<\frac{M}{m}}\max_{jm<l_1<l_2<(j+1)m} \frac{|\widecheck{\bm}_{l_1,l_2}|}{\sqrt{(l_2-l_1)\Delta_n}} < \cv (1+\iota_n)\right\} \\
&~~~+ \measure\left\{ \max_{0 \leq j<\frac{M}{m}}\max_{jm<l_1<l_2<(j+1)m} \frac{|\widehat{\bm}_{l_1,l_2}-\widecheck{\bm}_{l_1,l_2}|}{\sqrt{(l_2-l_1)\Delta_n}} \geq \cv\iota_n\right\} \nonumber
\end{align}
for some deterministic sequence of positive $\iota_n$ satisfying $\iota_n\rightarrow 0$ and $\cv\iota_n\rightarrow \infty$.

Consider the first term and let $\epsilon\in(0,1)$, we have
\begin{align*}
&~ \measure\left\{ \max_{0 \leq j<\frac{M}{m}}\max_{jm<l_1<l_2<(j+1)m} \frac{|\widecheck{\bm}_{l_1,l_2}|}{\sqrt{(l_2-l_1)\Delta_n}} < \cv (1+\iota_n)\right\} \\
&= \prod_{0 \leq j<\frac{M}{m}}\measure\left\{\max_{jm<l_1<l_2<(j+1)m} \frac{|\widecheck{\bm}_{l_1,l_2}|}{\sqrt{(l_2-l_1)\Delta_n}} < \cv(1+\iota_n) \right\}  \\
&= \left(1 - \measure\left\{\max_{0<l_1<l_2<m} \frac{|\widecheck{\bm}_{l_1,l_2}|}{\sqrt{(l_2-l_1)\Delta_n}} > \cv (1+\iota_n) \right\}\right)^{\frac{M}{m}}  \\
&\leq \left(1 - m\ARLH_a\cv(1+\iota_n)\phi(\cv(1+\iota_n))(1 - \epsilon)\right)^{\frac{M}{m}},
\end{align*}
where the last inequality holds for all large $\cv$ due to the result in (\ref{eqn:siegmundProp1}) with $\bar{c}$ chosen such that $m \lesssim w_n$ and thus $\ARLHc \lesssim \ARLH_a < \ARLH$, where $\ARLH_a := d(a)$ (see \cite{lai1999efficient} and \cite{chan2003saddlepoint}). Note that $\ARLH = \int_{0}^{\infty}x\nu^2(x)\dd x = \lim_{\overline{c}\to\infty}\ARLHc = \lim_{a\to\infty}\ARLH_a$. Then, because $1 - x \leq \exp(-x)$ for $x>0$, we have 
\begin{align*}
&~ \measure\left\{ \max_{0 \leq j<\frac{M}{m}}\max_{jm<l_1<l_2<(j+1)m} \frac{|\widecheck{\bm}_{l_1,l_2}|}{\sqrt{(l_2-l_1)\Delta_n}} < \cv (1+\iota_n)\right\} \\
&\leq \exp\left(- M\ARLH_a\cv(1+\iota_n)\phi(\cv(1+\iota_n))(1 - \epsilon)\right) \\
&\sim \exp\left(- M\ARLH_a\cv\phi(\cv)(1 - \epsilon)\right) \\
&= \exp\left(- r(1 - \epsilon)\right).
\end{align*}

Turning to the second term, we show it is asymptotically negligible relative to the first term. Recalling that $m = \overline{c}\cv^2$, we apply Proposition~\ref{prop:consistency_BM} (note that $M\Delta_n\rightarrow 0$ by the assumption of the theorem) and have

\begin{align*}
\begin{split}
&~ \measure\left\{ \max_{0 \leq j<\frac{M}{m}}\max_{jm<l_1<l_2<(j+1)m} \frac{|\widehat{\bm}_{l_1,l_2}-\widecheck{\bm}_{l_1,l_2}|}{\sqrt{(l_2-l_1)\Delta_n}} \geq \cv\iota_n\right\} \\
&\leq \sum_{j=0}^{M/m}\measure\left\{ \max_{jm<l_1<l_2<(j+1)m} \frac{|\widehat{\bm}_{l_1,l_2}-\widecheck{\bm}_{l_1,l_2}|}{\sqrt{(l_2-l_1)\Delta_n}} \geq \cv\iota_n\right\}  \\
&~ \leq C_0M\left(\frac{\sqrt{m}(l_n\delta_n + \sqrt{\Delta_n})}{\xi\iota_n} + \Delta_n^{1-\varpi}\right) 
\leq C_0M\left(\frac{(l_n\delta_n+\sqrt{\Delta_n})}{\iota_n}+\Delta_n^{1-\varpi}\right). 
\end{split}
\end{align*}
If we pick $\iota_n$ such that $\iota_n\xi\rightarrow \infty$ and in addition $\left[(l_n\delta_n + \sqrt{\Delta_n})/\iota_n + \Delta_n^{1-\varpi}\right]/[\cv\phi(\cv)]\rightarrow 0$ (which is possible due to the restriction in the theorem) and noting that $\Delta_n^{1-\varpi}/\sqrt{\Delta_n} \to 0$ (because $\varpi < 1/2$), then 
\begin{align*}
\begin{split}
& \measure\left\{ \max_{0 \leq j<\frac{M}{m}}\max_{jm<l_1<l_2<(j+1)m} \frac{|\widehat{\bm}_{l_1,l_2} - \widecheck{\bm}_{l_1,l_2}|}{\sqrt{(l_2-l_1)\Delta_n}} \geq \cv\iota_n\right\}  \\
& \bigg/\measure\left\{ \max_{0 \leq j<\frac{M}{m}}\max_{jm<l_1<l_2<(j+1)m} \frac{|\widecheck{\bm}_{l_1,l_2}|}{\sqrt{(l_2-l_1)\Delta_n}} < \cv(1+\iota_n)\right\}\rightarrow 0.
\end{split}
\end{align*}

Altogether, therefore, we have, for $\epsilon\in(0,1)$, that
\begin{align}\label{eqn:thm1_result_oneside}
\measure\left\{\stoptime^{\rm w} > M\right\} \leq \exp\left( - M\ARLH_a\cv\phi(\cv)(1 - \epsilon)\right) = \exp\left( - r(1 - \epsilon)\right).
\end{align}

We next proceed with bounding $\measure\left\{\stoptime^{\rm w} \leq M\right\}$. We have
\begin{align}  \label{eqn:thm1_expand_anotherside}
\measure\left\{\stoptime^{\rm w} \leq M\right\} 
\leq 
&~ \measure\left\{ \max_{0 \leq j<\frac{M}{m}}\max_{jm<l_1<l_2<(j+1)m} \frac{|\widehat{\bm}_{l_1,l_2}|}{\sqrt{(l_2-l_1)\Delta_n}} > \cv \right\}  \\
&~ + \measure\left\{ \max_{0 \leq j<\frac{M}{m}}\max_{jm<l_1<(j+1)m, (j+1)m<l_2<(j+1)m+w_n} \frac{|\widehat{\bm}_{l_1,l_2}|}{\sqrt{(l_2-l_1)\Delta_n}} > \cv \right\}, \nonumber
\end{align}
where in the second term we need to consider only the maximum over $(j+1)m<l_2<(j+1)m+w_n$ for each $j$ due to the limited window size. Let $\epsilon>0$. For the first term, using similar trick to the one for bounding $\measure\left\{\stoptime^{\rm w} > M\right\}$ above, we can show
\begin{align*} 
& \measure\left\{ \max_{0 \leq j<\frac{M}{m}}\max_{jm<l_1<l_2<(j+1)m} \frac{|\widehat{\bm}_{l_1,l_2}|}{\sqrt{(l_2-l_1)\Delta_n}} > \cv \right\} \\
&= 1 - \measure\left\{ \max_{0 \leq j<\frac{M}{m}}\max_{jm<l_1<l_2<(j+1)m} \frac{|\widehat{\bm}_{l_1,l_2}|}{\sqrt{(l_2-l_1)\Delta_n}} \leq \cv \right\} \\
&\leq 1 - \exp\left(-M\ARLH_a\cv\phi(\cv)(1 + \epsilon)\right) \\
&= 1 - \exp\left(-r(1 + \epsilon)\right).
\end{align*}
If we can further show that the second term in (\ref{eqn:thm1_expand_anotherside}) is asymptotically negligible, relative to  $\exp\left( - r(1 + \epsilon)\right)$, we will have 
\begin{align} \label{eqn:thm1_result_anotherside}
\measure\left\{\stoptime^{\rm w} > M\right\} \geq \exp\left( - r(1 + \epsilon)\right),
\end{align}
which, together with (\ref{eqn:thm1_result_oneside}), implies
\begin{align*}
\measure\left\{\stoptime^{\rm w} \leq M\right\} \sim 1 - \exp\left( - r\right),
\end{align*}
that $\stoptime^{\rm w}(\cv)$ is asymptotically (as $\cv\to\infty$) exponentially distributed with mean $1/(\ARLH_a\cv\phi(\cv))$.

For the second term in (\ref{eqn:thm1_expand_anotherside}), first we note that 
\begin{align*}
&\measure\left\{ \max_{0<j<\frac{M}{m}}\max_{jm<l_1<(j+1)m, (j+1)m<l_2<(j+1)m+w_n} \frac{|\widehat{\bm}_{l_1,l_2}|}{\sqrt{(l_2-l_1)\Delta_n}} > \cv \right\}  \\
\leq~ &\measure\left\{ \max_{0<j<\frac{M}{m}}\max_{jm<l_1<(j+1)m, (j+1)m<l_2<(j+1)m+w_n} \frac{|\widecheck{\bm}_{l_1,l_2}|}{\sqrt{(l_2-l_1)\Delta_n}} > \cv(1-\iota_n) \right\}  \\
& + \measure\left\{ \max_{0<j<\frac{M}{m}}\max_{jm<l_1<(j+1)m, (j+1)m<l_2<(j+1)m+w_n} \frac{|\widehat{\bm}_{l_1,l_2} - \widecheck{\bm}_{l_1,l_2}|}{\sqrt{(l_2-l_1)\Delta_n}} \geq \cv\iota_n \right\},
\end{align*}
for the same sequence $\iota_n$ as above. Again, by Proposition~\ref{prop:consistency_BM} and the condition that $[(\delta_nl_n+\sqrt{\Delta_n})/\iota_n+\Delta_n^{1-\varpi}]/[\cv\phi(\cv)] \rightarrow 0$, the second probability on the right-hand-side of the above inequality is asymptotically negligible relative to the first one.

For the first of the two probabilities, we can show that it is less than $\epsilon$, provided $\bar{c}$ is large enough, for all large $\cv$ (i) by applying again the random field result in \citet{siegmund1988approximate} if $w_n$ is at the order of $\bar{c}^{1/2}m = \bar{c}^{3/2}\cv^2$ or (ii) by \citet[Lemma 2--8]{siegmund1995using} if $w_n$ is at a larger order (say, e.g., of $1/\cv\phi(\cv)$) or even through infinity. The latter case (ii) thus applies also to the non-window-limited stopping rule $\stoptime(\cv)$ (where $w_n = \infty$), as well as all above arguments. For this case we replace $\ARLH_a$ by $\ARLH$ since $\ARLH_a \to \ARLH$ as $a_n\to\infty$.  
\end{proof}

\bigskip

\begin{proof}[Proof of Theorem~\ref{thm:GLRw-CUSUM_alter}] 
Without loss of generality, we can set $\sigma_0 = 1$ and we do so in the rest of the proof in order to simplify exposition. For some deterministic sequence $\iota_n\downarrow 0$ and such that $\cv\iota_n\rightarrow \infty$ and $\sqrt{w_n}l_n\delta_n\Delta_n^{\varpi-1/2}/(\xi\iota_nT_n)\rightarrow 0$ (such a sequence is possible to pick due to the restriction on $w_n$ of the theorem), we have
\begin{align}
&\mathbb{P}\left(\stoptime^{\rm w}(\cv)>\widebar{T}_n\right)
\leq 
\mathbb{P}\left(\overline{\stoptime}^{\rm w}(\cv(1+\iota_n))>\widebar{T}_n\right) \nonumber \\
&~ + \mathbb{P}\left(\max_{0\leq j<\frac{\widebar{T}_n}{w_n}}\max_{jw_n<l_1<l_2<(j+1)w_n} \frac{|\widehat{\bm}_{l_1,l_2}-\overline{\bm}_{l_1,l_2}|}{\sqrt{(l_2-l_1)\Delta_n}} \geq \cv\iota_n\right),
\end{align}
where $\widebar{T}_n = \tau_n/\Delta_n+T_n$, $\overline{W}_{l_1,l_2} = \frac{1}{\sigma_0}\sum_{i = l_1+1}^{l_2} \Delta_i^n\logs\indicator_{\{|\Delta_i^n\logs|<\zeta\Delta_n^\varpi\}}$, and $\overline\stoptime^{\rm w}(\cv)$ is the stopping time associated with $\overline{W}_{l_1,l_2}$. For the second of the two probabilities on the right-hand side of the above inequality, we can argue as in the proof of Proposition~\ref{prop:consistency_BM} and given the condition of the theorem conclude that it is asymptotically negligible. Hence, it suffices to consider the behavior of $\overline\stoptime^{\rm w}(\cv)$ in what follows. 

Next, we denote  $i_n = \inf\{i=1,...,n: i\Delta_n\geq \tau_n\}$ and introduce
\begin{equation}
K_n = \inf\left\{i=0,1,...: \left|c\int_{(i+i_n-1)\Delta_n}^{(i+i_n)\Delta_n}(s-\tau_n)^{-\vartheta}ds\right|<\zeta\sqrt{\Delta_n}\right\}.
\end{equation}
For $n$ large, we have $K_n\asymp \Delta_n^{\frac{1/2-\vartheta}{\vartheta}}$. We denote the stopping time associated with a discretized process constructed from the increments $\frac{1}{\sigma_0}\Delta_i^n\logs\indicator_{\{|\Delta_i^n\logs|<\zeta\Delta_n^\varpi\}}$, for $i\geq i_n+K_n$, with $\overline\stoptime^{\rm w}_0(\cv)$. Since $w_n/K_n\rightarrow 0$, we have $\overline{\stoptime}^{\rm w}(\cv)>\widebar{T}_n$ $\Longrightarrow$ $\overline\stoptime^{\rm w}_0(\cv)>T_n$ and hence it suffices to bound the probability $\mathbb{P}\left(\overline\stoptime^{\rm w}_0(\cv)>T_n\right)$.

For the increment of $H$, we have
\begin{align*}
\Delta_i^{n}\pn 
&= \pn_{i\Delta_n} - H_{(i-1)\Delta_n} \\
&= \frac{c}{1-\vartheta}\left[(i\Delta_n-\tau_n)^{1-\vartheta} - ((i-1)\Delta_n-\tau_n)^{1-\vartheta}\right] \\
&\sim c\left(i-\frac{\tau_n}{\Delta_n}\right)^{-\vartheta}\Delta_n^{1-\vartheta}, {\rm ~~for~large~} i> \frac{\tau_n}{\Delta_n},
\end{align*}
and we note that the sequence of $\Delta_i^nH$ is decreasing in absolute value (with each element of the same sign), for $i > i_n$. Also, for $i>i_n+K_n$, we have that $|\Delta_i^nH-\Delta_{i-1}^nH|\leq \zeta\sqrt{\Delta_n}$ for $n$ large enough. Furthermore, given the behavior of maxima of normal i.i.d.\ variables with zero mean, see for instance Example 3.3.29 in \cite{embrechts2013modelling}, we have that 
\begin{equation}
\mathbb{P}\left(\max_{i=1,...,\lfloor \widebar{T}_n \rfloor}|\Delta_i^nW|>\log(1/\Delta_n)\sqrt{\Delta_n}\right)\rightarrow 0.
\end{equation}
Hence, it suffices to consider only the set on which $\max_{i=1,...,\lfloor \widebar{T}_n \rfloor}|\Delta_i^nW|\leq \log(1/\Delta_n)\sqrt{\Delta_n}$. On this set, since  $|\Delta_i^nH-\Delta_{i-1}^nH|\leq \zeta\sqrt{\Delta_n}$ for $i>i_n+K_n$ and $n$ sufficiently large, we have that the truncation is not triggered for $i>i_n+K_n$. As a result, we only need to bound 
$\mathbb{P}\left(\check\stoptime^{\rm w}_0(\cv)>T_n\right)$ with $\check\stoptime_0^{\rm w}(\cv)$ being the stopping time associated with $\check{W}_{l_1,l_2} = \sum_{i = i_n+K_n+l_1+1}^{i_n+K_n+l_2}(b_0\Delta_n+\Delta_i^nW + \Delta_i^nH)$.

We can then pick $M_n= K_n+\log(\cv)$. 
In this case, $\min_{i> \tau_n/\Delta_n}|\Delta_i^nH|>|\mu_{M_n}|$, where we set $\mu_{M_n} = cM_n^{-\vartheta}\Delta_n^{1-\vartheta}$. Note that, given the definition of $K_n$, $\sup_{i=i_n+K_n,...,i_n+M_n}|\Delta_i^nH-\mu_{M_n}|\leq C\sqrt{\Delta_n}$, for some random variable $C$. Hence, 
\begin{align}
&\mathbb{P}\left(\check\stoptime^{\rm w}_0(\cv)>T_n\right)
\leq 
\mathbb{P}\left(\dot\stoptime^{\rm w}_0(\cv)>T_n\right) \nonumber \\
&~ + \mathbb{P}\left(\sup_{i=i_n+K_n,...,i_n+M_n}|\Delta_i^nH+b_0\Delta_n-\mu_{M_n}| \geq \cv\iota_n\right),
\end{align}
for $\dot\stoptime^{\rm w}_0(\cv)$ being the stopping time associated with $\dot{W}_{l_1,l_2} = \sum_{i = i_n+K_n+l_1+1}^{i_n+K_n+l_2}(\Delta_i^nW + \mu_i^n)$ with $\mu_i^n = \mu_{M_n}$ for $i=i_n+K_n,...,i_n+M_n$ and $\mu_i^n = \Delta_i^nH$ otherwise.

Finally, we denote
\begin{align} 
\widetilde\stoptime(\cv) := \inf\left\{l > 0: \max_{0<k<l} \frac{|\widetilde{\bm}_{k,l}|}{\sqrt{(l-k)\Delta_n}} > \cv \right\},
\end{align}
where $\widetilde{\bm}_{l_1,l_2}$ is the random walk with $N(\mu_{M_n}\Delta_n^{1/2},\Delta_n)$ increments on the index interval $[l_1,l_2]$, for $l_1, l_2 \in \SN_+$. Note that $\dot\stoptime^{\rm w}_0(\cv) = \widetilde\stoptime^{\rm w}(\cv)$ if $\dot\stoptime^{\rm w}_0(\cv)<M_n$. Hence, we will be done if we can show that $\mathbb{P}\left(\widetilde\stoptime^{\rm w}(\cv)< M_n\right)\rightarrow 1$.


This is what we do next. The rest of our proof is based on the \textit{nonlinear renewal theory} of \cite{lai1979nonlinear} and the argument for CUSUM-type statistic by \cite{dragalin1994optimality}.\footnote{One could also consider a very similar nonlinear renewal theory in \citet{hagwood1982expansion} (and verify the conditions therein).} To keep notation simple, we use $\mu$ for $\mu_{M_n}$ below.

We decompose the \textit{nonlinear} (square of the) statistic term into a summation of a linear term $\mu S_l$ and an additional term $V_l'$, as
\begin{align*}
  \max_{0<k<l} \frac{(\widetilde{\bm}_{k,l})^2}{2(l-k)\Delta_n}
&= \max_{0<k<l} \left\{ \mu(S_l - S_k) + \frac{(\widetilde{\bm}_{k,l}/\sqrt{\Delta_n} - (l-k)\mu)^2}{2(l-k)} \right\} \\
&= \mu S_l + V_l'
\end{align*}
where 
\begin{align*}
S_l &:= \widetilde{\bm}_{0,l}/\sqrt{\Delta_n} - \frac{1}{2}l\mu \\
V_l' &:= - \min_{0<k<l}\mu S_k + \max_{0<k<l}\left\{ \frac{(\widetilde{\bm}_{k,l}/\sqrt{\Delta_n} - (l-k)\mu)^2}{2(l-k)} \right\}.
\end{align*}
And the stopping rule can be written as 
\begin{align} 
\widetilde\stoptime(\cv) = \inf\left\{l > 0: \mu S_l + V_l' > \cv^2/2 \right\} 
\end{align}

For term $V_l'$, we first introduce $V_l''$, defined as 
\begin{align}
V_l'' &:= - \min_{0<k<l}\mu S_k + V_l, \\
V_l &:= (\widetilde{\bm}_{0,l}/\sqrt{\Delta_n} - l\mu)^2/(2l).
\end{align} 
The sequence $\{V_l''\}_{l \geq 1}$ is slowly changing because, as $l\to\infty$, $\min_{0<k<l}S_k$ converges weakly to $\min_{k>0}S_k$, and the second term $V_l$ weakly converges to $\chi_1^2/2$ random variable ($\chi^2$-distributed r.v.\ with $1$ degree of freedom). Next, as argued in 
\cite[Lemma~1]{dragalin1994optimality}, 
$V_l' = V_l''$ on a set $\mathcal{A}_\cv$ whose $\measure_{\mu}$-probability convergences to $1$ exponentially fast as $\cv\to\infty$ with the rate $1-O(e^{-\lambda\cv})$ for some $\lambda>0$. Hence, 
\begin{align}
\E[V_{\widetilde\stoptime}';\mathcal{A}_\cv] \to \E[V_{\widetilde\stoptime}''] = -\mu\E[\min_{k>0}S_k] + \frac{1}{2}.
\end{align}

Noting that $\mu S_l$ corresponds to a random walk with $N(\mu^2/2,\mu^2)$ increments, \citet[Theorem~3]{lai1979nonlinear} implies that
\begin{align} \label{eqn:EDD_expectationexpansion}
\frac{1}{2}\mu^2\E_0[\widetilde\stoptime(\cv)] = \frac{1}{2}\cv^2 - 0 - \left[-\mu\E[\min_{k>0}S_k] + \frac{1}{2}\right] + \frac{\mu^2\E[S_{\tau_{+}}^2]}{2\mu\E[S_{\tau_{+}}]} + o(1),
\end{align}
that is, 
\begin{align} \label{eqn:EDD_result}
\E_0[\widetilde\stoptime(\cv)] = \frac{\cv^2 - 1}{\mu^2} + \frac{1}{\mu}\left[2\E[\min_{k\geq 0} S_{k}] + \frac{1}{\mu}\frac{\E[S_{\tau_{+}}^2]}{\E[S_{\tau_{+}}]}\right] + o(1),
\end{align}
for which an approximation provided by \cite{siegmund1995using} is
\begin{align} \label{eqn:EDD_result_approx}
\E_0[\widetilde\stoptime(\cv)] \approx \frac{\cv^2-3}{\mu^2} + \frac{4\rho}{\mu},
\end{align}
where $\rho \approx 0.583$.

Now we verify the conditions of \citet[Theorem~3]{lai1979nonlinear} which apply to our case
\begin{itemize}
\item[(A)] For some $\delta>0$, $\measure\{\widetilde\stoptime\leq\delta\cv^2\} = o(\cv^{-2})$ as $\cv\to\infty$,
\item[(B)] Assume that, $\E\,|x_1|^p < \infty$ for some $p>2$,
\item[(C)] $V_l$ converges in distribution to a random variable $V$,
\item[(D)] $x_1$ is non-lattice and there exists $1/2 < \alpha \leq 1$ such that $\E\,|x_1|^{2/\alpha} < \infty$,
\end{itemize}
where $x_i := \mu(\Delta_i^n\widetilde{\bm}/\sqrt{\Delta_n} - \mu/2)$.
In particular, Condition (A) holds as long as the distribution of $\Delta_i^n\widetilde{\bm}$ belongs to exponential density, see discussion in \cite[Section 4]{lai1979nonlinear}.\footnote{Condition (A) is also used in \citet{hagwood1982expansion}, which is their condition (11). The paper shows that it is satisfied when our Condition (B) holds.} Condition (B) can imply (14), (16) and (17) in \cite{lai1979nonlinear} for their Theorem 3, see their Proposition 1, and it is trivially holds in our case since $x_i$'s are normally distributed. Condition (C) is immediate by the weak convergence of $V_l$. Condition (D) is also met trivially in our case because $\widetilde{\bm}$ is a random walk with normal increments.

The result above holds also for the window-limited version of $\widetilde\stoptime(\cv)$ which replaces $0<k<l$ by $l-w_n<k<l$ for an increasing window $w_n$ being chosen such that $w_n \sim a_n\cv^2$ with $a_n$ being some deterministic sequence converging to $a\in (2/c^2,\infty]$, see \citet[Theorem 2]{xie2013sequential} or \citet[Theorem 1]{lai1999efficient}.

We are now ready to show $\mathbb{P}\big(\widetilde\stoptime^{\rm w}(\cv)< M_n\big)\rightarrow 1$. Recall that $\mu = \mu_{M_n} \sim cM_n^{-\vartheta}\Delta_n^{1/2-\vartheta}$, by plugging-in the $\mu_{M_n}$ value we have
\begin{align} \label{eqn:EDD_untruncated}
\E_0[\widetilde\stoptime(\cv)]
&\approx \frac{\cv^2 - 3}{c^2}(M_n\Delta_n)^{2\vartheta}\Delta_n^{-1} + \frac{4\rho}{c}(M_n\Delta_n)^{\vartheta}\Delta_n^{-1/2}.
\end{align}
Now, since $M_n\asymp \Delta_n^{(1/2-\vartheta)/\vartheta}$, we have that $\E_0[\widetilde\stoptime(\cv)]\leq 2\cv^2/c^2$ for $n$ sufficiently high. From here, since $\cv\Delta_n^{\iota}\rightarrow 0$, we have that $\mathbb{P}\big(\widetilde\stoptime^{\rm w}(\cv)< M_n\big)\rightarrow 1$.

\end{proof}

\end{document}